\documentclass[journal]{IEEEtran}

\usepackage{cite}
\usepackage{amsmath,amssymb,amsfonts}
\usepackage{algorithmic}

\ifCLASSINFOpdf
  \usepackage[pdftex]{graphicx}
  \graphicspath{{Figures/}}
 \DeclareGraphicsExtensions{.pdf,.jpeg,.png}
\else
   \usepackage[dvips]{graphicx}
\fi

\usepackage{amsthm}
\usepackage{epsfig}
\usepackage{color}
\usepackage{lettrine}
\usepackage{float}
\usepackage{lipsum}
\usepackage{bm}
\usepackage{mathtools}
\usepackage{bbm}
\usepackage{suffix}
\usepackage[T1]{fontenc}
\usepackage[usenames,dvipsnames]{xcolor}
\definecolor{darkpastelgreen}{rgb}{0.01, 0.75, 0.24}
\usepackage[colorlinks=true,citecolor = darkpastelgreen]{hyperref} 
\usepackage[nameinlink,noabbrev]{cleveref}
\usepackage{textcomp}
\usepackage{stackengine}
\usepackage[ruled,norelsize]{algorithm2e}
\def \treq {\stackrel{\tiny \Delta}{=}}
\usepackage{multicol}
\usepackage{euler}
\usepackage{caption}
\usepackage{subcaption}
\newcommand{\Q}{\ensuremath{\mathrm{Q}}}
\newcommand{\Qi}{\ensuremath{\mathrm Q}^{-1}}
\newcommand{\e}{\bar{e}}
\newcommand{\h}{\mathbf{h}}
\renewcommand{\H}{\mathbf{H}}
\newcommand{\C}{\ensuremath{\mathbb C}}
\newcommand{\g}{\mathbf{g}}
\newcommand{\E}{\ensuremath{\mathbb E}}
\renewcommand{\Pr}{\ensuremath{\mathbb P}}
\newcommand{\R}{\ensuremath{\mathbb R}}

\newcommand{\CN}{\ensuremath{\mathcal {CN}}}
\newcommand{\N}{\ensuremath{\mathcal {N}}}
\newcommand{\I}{\ensuremath{\mathbf I}}

\newcommand{\x}{\ensuremath{\mathbf x}}
\newcommand{\w}{\ensuremath{\mathbf w}}
\newcommand{\W}{\ensuremath{\mathbf W}}
\newcommand{\bgamma}{\ensuremath{\bar{\gamma}}}
\newcommand{\bdelta}{\ensuremath{\bar{\delta}}}
\newcommand{\Gam}[2]{\mathcal{G}\left(#1, #2\right)}
\newcommand{\Exp}[1]{\mathcal{E}\left(#1\right)}

\newcommand{\U}{\ensuremath{\mathbf U}}
\renewcommand{\v}{\ensuremath{\mathbf v}}
\newcommand{\X}{\ensuremath{\mathbf X}}

\renewcommand{\b}{\ensuremath{\mathbf b}}
\renewcommand{\o}{\ensuremath{\mathbf 0}}
\newcommand{\Prob}[1]{(\textbf{P}{#1})}
\newcommand{\cst}[1]{\textbf{C}{#1}}
\newcommand{\V}{\ensuremath{\mathbb{V}}}
\newcommand*\diff{\mathop{}\!\mathrm{d}}

\def \treq {\stackrel{\tiny \Delta}{=}}

\newtheorem{prop}{Proposition}
\newtheorem{cor}{Corollary}
\newtheorem{remark}{Remark}
\newtheorem{lemma}{Lemma}

\usepackage{scalerel}
\usepackage{tikz}
\usetikzlibrary{svg.path}

\definecolor{orcidlogocol}{HTML}{A6CE39}
\tikzset{
  orcidlogo/.pic={
    \fill[orcidlogocol] svg{M256,128c0,70.7-57.3,128-128,128C57.3,256,0,198.7,0,128C0,57.3,57.3,0,128,0C198.7,0,256,57.3,256,128z};
    \fill[white] svg{M86.3,186.2H70.9V79.1h15.4v48.4V186.2z}
                 svg{M108.9,79.1h41.6c39.6,0,57,28.3,57,53.6c0,27.5-21.5,53.6-56.8,53.6h-41.8V79.1z M124.3,172.4h24.5c34.9,0,42.9-26.5,42.9-39.7c0-21.5-13.7-39.7-43.7-39.7h-23.7V172.4z}
                 svg{M88.7,56.8c0,5.5-4.5,10.1-10.1,10.1c-5.6,0-10.1-4.6-10.1-10.1c0-5.6,4.5-10.1,10.1-10.1C84.2,46.7,88.7,51.3,88.7,56.8z};
  }
}

\newcommand\orcidicon[1]{\href{https://orcid.org/#1}{\mbox{\scalerel*{
\begin{tikzpicture}[yscale=-1,transform shape]
\pic{orcidlogo};
\end{tikzpicture}
}{|}}}}

\makeatletter
\def\ps@IEEEtitlepagestyle{%
  \def\@oddfoot{\mycopyrightnotice}%
  \def\@oddhead{\hbox{}\@IEEEheaderstyle\leftmark\hfil\thepage}\relax
  \def\@evenhead{\@IEEEheaderstyle\thepage\hfil\leftmark\hbox{}}\relax
  \def\@evenfoot{}%
}
\def\mycopyrightnotice{%
  \begin{minipage}{\textwidth}
  \centering \scriptsize
  \copyright~2024 IEEE. Personal use of this material is permitted.  Permission from IEEE must be obtained for all other uses, in any current or future media, including reprinting/republishing this material for advertising or promotional purposes, creating new collective works, for resale or redistribution to servers or lists, or reuse of any copyrighted component of this work in other works.
  \end{minipage}
}
\makeatother
\begin{document}
\bstctlcite{IEEEexample:BSTcontrol}

    \title{Performance Analysis of Finite Blocklength Transmissions Over Wiretap Fading Channels: An Average Information Leakage Perspective}

\author{Milad Tatar Mamaghani\textsuperscript{\orcidicon{0000-0002-3953-7230}},~\IEEEmembership{Member,~IEEE}, Xiangyun Zhou\textsuperscript{\orcidicon{0000-0001-8973-9079}},~\IEEEmembership{Fellow,~IEEE}, Nan Yang\textsuperscript{\orcidicon{0000-0002-9373-5289}},~\IEEEmembership{Senior Member,~IEEE},
A. Lee Swindlehurst\textsuperscript{\orcidicon{0000-0002-0521-3107}},~\IEEEmembership{Fellow,~IEEE}, and H. Vincent Poor\textsuperscript{\orcidicon{0000-0002-2062-131X}},~\IEEEmembership{Life Fellow,~IEEE}
\thanks{This work was supported in part by the Australian Research Council’s Discovery Projects under Grant DP220101318 and in part by the U.S. National Science Foundation under Grants CNS-2128448 and ECCS-2335876.}
\thanks{Preliminary results of this work have been accepted for presentation at the IEEE International Conference on Communications, 9-13 June 2024, Denver, CO, USA \cite{mamaghani2024icc}.}
\thanks{M. Tatar Mamaghani, X. Zhou, and N. Yang are with the School of Engineering, The Australian National University, Canberra, ACT 2601, Australia (email: \href{mailto:milad.tatarmamaghani@anu.edu.au}{\textcolor{black}{milad.tatarmamaghani@anu.edu.au}}; \href{mailto:xiangyun.zhou@anu.edu.au}{\textcolor{black}{xiangyun.zhou@anu.edu.au}}; \href{mailto:nan.yang@anu.edu.au}{\textcolor{black}{nan.yang@anu.edu.au}}).}
\thanks{A. Lee Swindlehurst is with the Center for Pervasive Communications and
Computing, Henry Samueli School of Engineering, University of California, Irvine, CA 92697, USA (email: \href{mailto:swindle@uci.edu}{\textcolor{black}{swindle@uci.edu}}).}
\thanks{H. Vincent Poor is with the Department
of Electrical and Computer Engineering, Princeton University, Princeton,
NJ 08544, USA (email: \href{mailto:poor@princeton.edu}{\textcolor{black} {poor@princeton.edu}}).}
}

\maketitle

\begin{abstract}
Physical-layer security (PLS) is a promising technique to complement more traditional means of communication security in beyond-5G wireless networks. However, studies of PLS are often based on ideal assumptions such as infinite coding blocklengths or perfect knowledge of the wiretap link's channel state information (CSI). In this work, we study the performance of finite blocklength (FBL) transmissions using a new secrecy metric --- the average information leakage (AIL). We evaluate the exact and approximate AIL with Gaussian signaling and arbitrary fading channels, assuming that the eavesdropper's instantaneous CSI is unknown. We then conduct case studies that use artificial noise (AN) beamforming to analyze the AIL in both Rayleigh and Rician fading channels. The accuracy of the analytical expressions is verified through extensive simulations, and various insights regarding the impact of key system parameters on the AIL are obtained. Particularly, our results reveal that allowing a small level of AIL can potentially lead to significant reliability enhancements. To improve the system performance, we formulate and solve an average secrecy throughput (AST) optimization problem via both non-adaptive and adaptive design strategies. Our findings highlight the significance of blocklength design and AN power allocation, as well as the impact of their trade-off on the AST.
\end{abstract}


\IEEEpeerreviewmaketitle

\section{Introduction}
\IEEEPARstart{R}{ecently}, there has been significant interest in new ways of providing  communication security, particularly at the physical layer. Unlike traditional security methods that rely solely on cryptographic algorithms at higher layers of the protocol stack, physical-layer security (PLS) strategies take advantage of the characteristics and inherent randomness of the physical transmission medium itself such as channel impairments, noise, and smart signaling to protect the confidentiality and integrity of wireless transmissions \cite{Poor2017}. PLS can improve wireless communication secrecy by providing an extra layer of resistance against attacks, reducing the reliance on key-based encipherment, and eliminating the need for complex key-exchange protocols. This is particularly important considering the resource restrictions of communication nodes in the Internet of Things (IoT) \cite{Mukherjee2015, Zou2016, Wang2019c}. The development of conventional PLS technologies has typically relied on the concept of secrecy capacity, which refers to the maximum achievable secrecy rate that guarantees both reliability and confidentiality over a wiretap channel. Wyner in his seminal work \cite{Wyner1975} established that by employing appropriate secrecy coding techniques, such as wiretap coding, it is possible to minimize the error probability and information leakage to an arbitrarily low level. This goal is achieved by selecting a data rate below the secrecy capacity and mapping information to asymptotically long codewords that approach infinity in length.

Nevertheless, emerging scenarios in beyond-5G wireless systems such as massive machine-type communication (mMTC) and ultra-reliable low-latency communication (uRLLC) are expected to support new types of data traffic characterized by finite blocklength (FBL) packets. Notably, the size of FBL packets can potentially be as little as tens of bytes compared to several kilobytes in conventional communication systems,  to satisfy the broader communication requirements of diversified applications such as the Internet of Vehicles (IoV) and the Industrial Internet of Things (IIoT) \cite{Akyildiz2020, Zhu2022, Bennis2018, Bockelmann2016}. While leveraging FBL communication minimizes end-to-end transmission latency due to reduced transmission blocklengths or the number of channel uses, it generally comes with a severe reduction in channel coding gain, making it challenging to ensure communication reliability. Furthermore, security issues impose imperative challenges in FBL communication. However, traditional PLS schemes (e.g., see \cite{TatarMamaghani2018, TatarMamaghani2022, mamaghani2021joint} and references therein) are no longer applicable to FBL communication scenarios or lead to significantly sub-optimal designs, due to the assumption of codewords with infinite blocklength (IBL). Thus, FBL transmission has a detrimental impact on both reliability and security, and it is of critical importance to carefully understand and design PLS schemes tailored to the specific requirements of the FBL regime.

\vspace{-3mm}
\subsection{Related Works and Motivation}
Some research has been conducted on quantifying the channel coding rate of FBL communication systems from an information-theoretic perspective. Polyanskiy \textit{et al.} in \cite{Polyanskiy2010} investigated the maximal channel coding rate attainable for a given blocklength and the decoding error probability for general classes of communication channels in the FBL regime. This work motivated the research community to further explore the characterization of non-asymptotic achievable rate regions in different scenarios and determine the practical impacts of wireless  FBL communication. In \cite{Mary2016c}, the authors studied the effects of the non-asymptotic coding rate for fading channels with no channel state information (CSI) at the transmitter, and then analyzed the good-put and energy efficiency over additive white Gaussian noise (AWGN) channels. The authors in \cite{Thapliyal2022} examined the average block error rate for an unmanned aerial vehicle (UAV)-assisted communication system operating in the FBL regime over Rician fading channels, and formulated a blocklength minimization problem with reliability constraints for the considered system. However, it is crucial to highlight that the aforementioned studies have not addressed the important issue of confidentiality in FBL communication, where sensitive information must often be transmitted \cite{Feng2021}.

Recent studies have considered various FBL communication systems from the perspective of PLS while adopting different security strategies. The authors in \cite{Li2022} compared secure transmission rates in FBL communication with queuing delay requirements under different CSI assumptions, and found that having perfect knowledge of the instantaneous CSI of the eavesdropper only provides a marginal gain in the secrecy rate for the high signal-to-noise ratio (SNR) regime. The authors in \cite{Wang2019e} studied secure FBL communication for mission-critical IoT applications with an external eavesdropper, proposed an analytical framework to approximate the average achievable secrecy throughput with FBL coding, and determined the optimal blocklength for maximizing the secrecy throughput. In \cite{Ren2020}, a resource allocation problem for secure delay-intolerant communication systems with FBL was studied, and the authors optimized both the channel bandwidth allocation to minimize power consumption while maintaining certain communication requirements.  The authors in \cite{Chen2020b} investigated the secrecy performance of cognitive IoT with low-latency requirements and security demands, addressing the spectrum scarcity problem. Specifically, they derived closed-form approximations for the secrecy throughput and investigated its impact on transmission delay and reliability. The studies \cite{Wang2020b,mamaghani2023globecom,mamaghani2023secure}  explored designing UAV-enabled secure communication systems with FBL to improve the average secrecy rate performance while meeting security and reliability requirements. In \cite{Feng2022}, the authors presented an analytical framework to approximate the average secrecy throughput with FBL, based on an outage probability defined according to simultaneously establishing the reliability and secrecy of FBL transmission. 

Despite the merits of these studies, the secrecy performance of FBL transmissions over fading channels based on the information leakage metric has been less studied in the literature, and most prior work has regarded it as a preset value that indicates the specified level of secrecy. However, in fading channels where the instantaneous CSI of the eavesdropper’s channel is often unknown, legitimate users do not have direct control over the amount of information leakage, which is essentially a random quantity depending on the eavesdropper's channel variation. Therefore, this paper is motivated by the need for a statistical measure of information leakage in the FBL regime.

\vspace{-3mm}

\subsection{Contributions}

To the best of the authors' knowledge, this is the first study which lays a foundation for secure communication in the FBL regime over \textit{fading channels} from the viewpoint of information leakage. The three-fold contributions of this study are detailed below.
\begin{itemize}
    \item 
    We explore \textit{average information leakage} (AIL), a new measure of secrecy performance, in a multi-antenna wiretap communication system operating in the FBL regime. Assuming the instantaneous CSI of the eavesdropper's channel is unavailable, we establish a simple statistical relationship between the conventional secrecy outage probability widely used in the IBL regime and the variational distance-based information leakage formulation in the FBL regime. This relationship reveals the inherent connections between these two security-oriented performance formulations.
    
    \item  We develop a generalized framework to evaluate the exact and approximate AIL performance over arbitrary fading channels with the adopted beamforming, and then perform case studies with an artificial noise (AN) beamforming strategy over both Rayleigh and Rician fading channels to carefully analyze the AIL, where closed-form expressions are obtained. Using extensive simulations, we confirm the accuracy of the theoretical results and explore the impact of various key system parameters.   
    Specifically, our findings demonstrate that permitting a small amount of AIL can significantly enhance reliability. Also, compared with Rayleigh fading, the LoS component in Rician fading is beneficial in decreasing AIL for a given reliability performance level.
    
    \item  Based on the analytical results, we formulate a design problem to maximize the average secrecy throughput using FBL transmission. We then optimize both the coding blocklength and power allocation using adaptive and non-adaptive schemes. In particular, for the non-adaptive scheme, we find that as the number of information bits increases, the optimal blocklength tends to grow linearly while less transmit power is generally allocated to AN signals to improve the overall AST.
\end{itemize}

\vspace{-3mm}  \subsection{Organization and Notation}
The remainder of this paper is organized as follows. Section~\ref{sec:sysmodel} introduces the system model and assumptions followed by the representation of the secrecy rate formulation in the FBL regime. In Section~\ref{sec:performance}, we develop a generalized framework for analyzing the secrecy performance of the system in terms of the AIL.  Section \ref{sec:case} examines case studies and derives closed-form expressions for the AIL performance. Next, we formulate a design problem and provide solutions in Section \ref{sec:problem} to find key system parameters that lead to improvement in the secure transmission performance. Simulation results and technical discussions are presented in Section \ref{sec:simulations}. Finally, conclusions are drawn in Section~\ref{sec:conclusion}.
   
\textit{Notation:} The following notation is used throughout the paper:
$\CN(\mu, \sigma^2)$ denotes a circularly symmetric complex normal distribution with mean $\mu$ and variance $\sigma^2$. In addition,
$\E\{X\}$ and $\V\{X\}$ denote the expectation and variance of the random variable (r.v.) $X$, respectively. Further, $f_{X}(\cdot)$ and $F_{X}(\cdot)$ represent the probability density function (PDF) and the cumulative distribution function (CDF) of the r.v. $X$, respectively. The operator $\lceil a \rceil$ indicates the smallest integer larger than or equal to $a$, while $\|\cdot\|$, $|\cdot|$, $(\cdot)^\dagger$ denote
the Euclidean norm, absolute value, and conjugate transpose, respectively. The real and imaginary parts of a variable $x$ are denoted respectively as $\Re(x)$ and $\Im(x)$, and $j$ indicates the imaginary unit with $j^2 = -1$. \textcolor{black}{The notation $\e$ represents Euler's number.} Moreover, $\mathbb{C}$, $\mathbb{R}^+$, and $\mathbb{Z}^+$ represent the sets of complex numbers, nonnegative real numbers, and nonnegative integer-valued numbers, respectively. Also, for a real valued $x$,  we define $[x]^+\treq\max\{x,0\}$. Further, $\I_k$ denotes the identity matrix of size $k$.  \textcolor{black}{Finally, the big-O notation $\mathcal{O}(\cdot)$ indicates a quantity that is of the same order as its argument.}

\section{System Model}\label{sec:sysmodel}

\begin{figure}[!t]
\centering
\includegraphics[width=0.9\columnwidth]{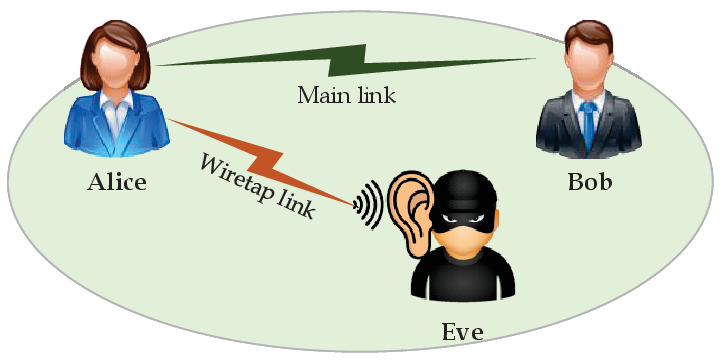}
\caption{System model for typical wiretap FBL communication in  IoT networks.}
\label{fig1}
\end{figure}

We consider a secure downlink IoT communication system, as depicted in Fig. \ref{fig1}, where a multi-antenna access point (Alice) serves an intended IoT device (Bob) through confidential FBL transmissions in the presence of a passive eavesdropper (Eve), who attempts to wiretap the ongoing legitimate communication. We assume that Alice is equipped with $k$ antennas, while both Bob and Eve have a single antenna, which forms a typical multiple-input single-output (MISO) scenario in IoT downlink applications. 

In this work, we assume quasi-static block fading channels, where the channel coefficients remain unchanged over the duration of one FBL packet within the channel coherence time, and are independent and identically distributed (i.i.d.) among different packets. 
Thus, we represent the reciprocal Alice-Bob channel (i.e., main link) and Alice-Eve channel (i.e., wiretap link)  by complex-valued channel vectors $\g_b\in\mathbb{C}^{1\times k}$  and $\g_e\in\mathbb{C}^{1\times k}$, where subscripts $b$ and $e$ refer to Bob and Eve, respectively. Notice that $\g_r,~r\in\{b,~e\}$, encapsulates both large-scale channel attenuation in terms of distance-dependent path loss and small-scale fading because of multi-path propagation. As a result, we have 
\begin{align}
    \g_r = \sqrt{\beta_r}{\h}_r,\quad r\in\{b,~e\}
\end{align}
where $\beta_r=\beta_0 d^{-\eta}_r$ accounts for large-scale attenuation with $\beta_0$, $d_r$, $\eta$ respectively indicating the channel power gain at the reference distance $1$m, the distance between Alice and receiver $r$, and the environmental path loss exponent. Furthermore, ${\h}_r\in\mathbb{C}^{1\times k}$, indicating the small-scale fading component, is a complex-valued normal random vector, whose entries have magnitudes that follow a normalized arbitrary fading distribution such as Rayleigh, Rician, etc., depending on the propagation environment.

We assume that the CSI for the main link is perfectly known; thus, Alice knows $\h_b$ and Bob is aware of his own CSI through channel training and reciprocity. We further assume that Eve can obtain her instantaneous CSI, however only statistical CSI for the wiretap link is available at the legitimate parties. Accordingly, the multi-antenna transmitter Alice should leverage an appropriate beamforming technique to combat eavesdropping and improve the confidentiality of the communication. Thus, we assume that Alice generates the unit-power transmit signal vector $\x$ such that $\E\{\|\x\|^2\}=1$, having applied an appropriate precoding strategy over the information-bearing signal $s$. Then, she sends the resultant $k\times 1$ signal $\x$ with a fixed transmit power $P$ over the fading channel $\g_r$. Accordingly,  the received signal at node $r$ for one channel use can be written as
\begin{align}\label{yj}
    y_r = \sqrt{P}  \g_r \x  + \xi_r, \quad r\in \{b,~e\} 
\end{align}
where $\xi_r \sim \CN(0, \sigma^2)$ accounts for the AWGN at node $r$. In the sequel,  we denote the received signal-to-noise ratio (SNR) at the intended node $r$  as $\gamma_r$, where $r\in \{b,~e\}$. Furthermore, the ratio $\frac{P}{\sigma^2}\treq \rho$ is referred to as the \textit{transmit SNR}.

\vspace{-3mm}  \subsection{Secure FBL communication}
The secrecy performance of wiretap systems has conventionally been evaluated according to Shannon's secrecy capacity criteria achieved by Wyner's wiretap coding in an IBL regime, where Bob can reliably recover the confidential message transmitted by Alice, provided that the transmission codeword rate $R_b$ (i.e., the transmission rate of both the intended confidential information and the additional redundant component) is not larger than the main channel capacity $C_b$, i.e., $R_b \leq C_b$. Recall that wiretap coset coding involves two rate parameters, namely the codeword rate $R_b$ and the secrecy rate $R_s$ (i.e., the rate of transmission of the intended confidential information), both of which can be designed to establish a positive rate difference or the so-called rate redundancy $R_e = R_b - R_s$, which reflects the cost of securing confidential message transmission against eavesdropping. 

One key feature of wiretap coding is the asymptotic perfect secrecy condition, i.e., given the availability of global CSI, zero information leakage is guaranteed asymptotically in an information-theoretic sense (regardless of Eve’s computing power) as the coding blocklength approaches infinity.  Nonetheless, when Eve's channel is not completely known, as assumed in this work, perfect secrecy is not always attainable and a secrecy outage may occur. Specifically, in an IBL regime, given knowledge of the main channel's CSI, the codeword rate $R_b$ can be set arbitrarily close to $C_b$ to ensure the maximum rate of transmission while not resulting in any decoding error at Bob. Thus, owing to the assumption that Eve's CSI is unknown to Alice and Bob, a secrecy outage event happens when the redundancy rate $R_e$ falls below the wiretap channel capacity $C_e$, i.e., $R_e < C_e$. Note that a secrecy outage is treated as a failure to achieve asymptotic perfect secrecy in the IBL regime.

\textcolor{black}{On the other hand, in the FBL regime, both a decoding error at Bob and information leakage to Eve may occur, leading to a further loss of communication reliability and secrecy compared to the IBL regime. The maximum achievable secrecy rate $R^*_s$ of Alice's FBL transmission for a message $W$ in the set $\mathcal{M} \treq \{1, ... , M\}$ with $m=\log_2 M$ bits and conveyed over a  blocklength $N$ can be defined as the maximal rate $\frac{m}{N}$ bits-per-channel-use (bpcu) such that a secrecy code $(m, N, \varepsilon, \delta)$ exists; i.e.
\begin{align}\label{exactfblsec}
    R^*_s \treq \sup\Big\{\frac{m}{N}: \exists (m, N, \varepsilon, \delta)~\text{secrecy code}\Big\}.
\end{align}
Note that the FBL secrecy coding rate given by \eqref{exactfblsec} incorporates Bob's desired maximum decoding error probability, characterized by $\varepsilon$, satisfying the constraint
\begin{align}
    \underset{m\in \mathcal{M}}{\max}~\Pr_{y_b|W=m}(g(y_b) \neq m) \leq \varepsilon,
\end{align}
where $\Pr_{y_b|W=m}(\cdot)$ is the conditional probability distribution of the received signal $y_b$ at Bob, given that Alice sends message $W\in\mathcal{M}$, and $g(\cdot)$ indicates Bob's adopted decoder to assign an estimate message $\hat{W}$ to $y_b$. Also, the secrecy constraint in terms of confidential information leakage to Eve, characterized by $\delta$\footnote{While the information leakage $\delta$ can theoretically range between 0 and 1, a smaller value will provide a stricter level of secrecy. In practice, this quantity should ideally be well below $0.1$ for acceptable secrecy performance. Nevertheless, $\delta$ is directly uncontrollable, particularly in scenarios with stochastic fading channels.}, is based on the notion of maximum total variation secrecy, expressed as 
\begin{align}
    S_{max}(W|y_e) \leq \delta,
\end{align}
where the quantity $S_{max}(W|y_e)$ is a statistical measure of independence between Alice's message $W$ and Eve's observation $y_e$. This measure is determined by calculating the variational distance between their joint and product distributions \cite{Yang2019}.
Note that the maximum total variation secrecy
is equivalent to the widely adopted secrecy metrics \textit{distinguishing secrecy} \cite{iwamoto2011security} and \textit{semantic secrecy} \cite{bellare2012semantic} (interested readers can refer to \cite{Yang2019, iwamoto2011security, bellare2012semantic} and references therein for detailed proofs and further discussion).}

We stress that calculating the precise secrecy capacity of FBL transmissions using \eqref{exactfblsec} is quite challenging. \textcolor{black}{To overcome this challenge, \cite{Yang2019} obtained a lower bound on $R^*_s$. By  dropping the big-O term $\mathcal{O}\left(\frac{\log_2 N}{N}\right)$ in this lower bound expression based on the asymptotic expansion and normal approximation, we approximate $R^*_s$ as
\begin{align}\label{sp_secrate}
    R^*_s &\approx \Big[ C_s(\gamma_b, \gamma_e) - \sqrt{\frac{V_b}{N}}\Qi\left(\varepsilon\right) -\sqrt{\frac{V_e}{N}}\Qi\left(\delta\right)\Big]^+,
\end{align}where $\Q^{-1}(x)$ is the inverse of the Gaussian Q-function defined as $\Q(x)=\int^{\infty}_{x}\frac{1}{\sqrt{2\pi}}\e^{-\frac{r^2}{2}}dr$, and $C_s$ is the secrecy capacity in the IBL regime given by
\begin{align}
    C_s(\gamma_b, \gamma_e) = \log_2(1+ \gamma_b) - \log_2(1+\gamma_e).
\end{align}}
The quantities $V_b$ and $V_e$ in \eqref{sp_secrate} represent the stochastic variation of Bob and Eve's channels, respectively, which can be expressed mathematically as
\begin{align}
    V_r = \log^2_2\e~\frac{\gamma_r(\gamma_r+2)}{(\gamma_r+1)^2},\quad r\in \{b,~e\}
\end{align}
\textcolor{black}{
Note that the penalty terms in \eqref{sp_secrate} are inversely proportional to $\sqrt{N}$, which indicates that $R^*_s$ asymptotically coincides with the IBL secrecy capacity $C_s$ as  $N$ approaches infinity.}

\section{Secrecy Performance}\label{sec:performance}
Our research focuses on the ergodic secrecy performance of FBL transmissions using the metric of information leakage, which is not well understood in the existing body of security literature. Moreover,  we aim to extend the fundamental results in \cite{Yang2019} by considering fading channels without knowledge of the instantaneous CSI of Eve's channel. Additionally, it is desirable to establish a simple and intuitive relationship in a statistical sense between the secrecy outage formulation widely used in the IBL regime and the variational distance-based information leakage formulation in the FBL regime. Our objective is to use this relationship to guide or at least simplify secure communication design in the FBL regime when dealing with fading channels.

\vspace{-3mm}  \subsection{Average Information Leakage}
Let us consider that Alice sends $m$  confidential information bits over $N$ channel uses in each  FBL transmission. Thus, the secrecy rate of the FBL transmission can be denoted by $R_s=\frac{m}{N}$, where the blocklength $N$ should be appropriately determined. 
Note that if the quality of the main channel is extremely low or the adopted secrecy rate $R_s$ is greater than the FBL secrecy capacity $R^*_s$, then an outage event would occur in which the desired secrecy and reliability performance for FBL transmission in terms of $(\varepsilon, \delta)$ could not be guaranteed. Thus, in order to meet such requirements during the transmission period, an on-off transmission policy is generally necessary to ensure that transmission takes place only when the main channel gain $\|\h_b\|^2$ is sufficiently strong; otherwise, Alice does not transmit\footnote{Note that an on-off threshold scheme may also consider the quality of Eve's channel \cite{He2013a}. For example, if the channel gain between Alice and Eve is strong, sending a message could greatly increase the likelihood of confidential information being leaked to Eve. However, as assumed in this work for a passive eavesdropping scenario, Alice does not have knowledge of $\gamma_e$, and thus she conducts FBL transmission solely based upon the feedback received from Bob.}.
In the case of no transmission, the information leakage to Eve is simply set to zero due to the lack of transmission. Henceforth, we assume that Alice does transmit and transmission occurs at a rate of $R_s$. 

In this work, we assume that Alice only knows the instantaneous CSI of Bob's channel but not that of Eve, which implies that the legitimate users have the ability to control/constrain the reliability performance ($\varepsilon$) but not the secrecy performance ($\delta$). Accordingly, we only set a requirement on reliability (i.e., $\varepsilon$ is fixed and known) and aim to determine the best secrecy performance that can be attained. To meet this aim, we set $R_s = R^*_s$, where $R^*_s$ is the achievable FBL secrecy rate given by \eqref{sp_secrate}, and then obtain the resulting information leakage $\delta$, formulated as 
\begin{align}\label{delta}
\delta =\Q\left(\sqrt{\frac{N}{V_e}}\left[C_s(\gamma_b, \gamma_e)- \sqrt{\frac{V_b}{N}}\Qi(\varepsilon) -\frac{m}{N}\right]\right).
\end{align}
By doing so, we obtain the amount of information leakage that would occur given the best possible coding strategy for transmitting the confidential information at the rate of $R_s$ with Bob's decoding error probability equal to $\varepsilon$. 
\textcolor{black}{
\begin{remark}
While not considered in this work, it is important to note that the scenario where the CSI of both the main link and the wiretap link is unknown sometimes reflects a more realistic and challenging environment, where uncertainties in both the main link and wiretap link need to be jointly addressed. We leave such a study to future work and, here, focus on the secrecy performance, given that the legitimate link is reliable enough to satisfy the required performance.
\end{remark}}

According to \eqref{delta}, it is evident that $\delta$ is a random quantity determined by the instantaneous SNR $\gamma_e$ of Eve's channel. We note that the sign of the argument of the Gaussian Q-function given by \eqref{delta} should be positive; otherwise, we have $\delta \geq 0.5$, which would greatly compromise the secrecy. Furthermore, since $\delta$ given in \eqref{delta} is conditioned on stochastic channel realizations of Eve and only statistical information about Eve's channel is assumed to be available, we average  the information leakage $\delta$ over all the realizations of $\gamma_e$ as
\begin{align}\label{delta_bar}
    \bar{\delta} &= \E\{\delta |~\h_b\}\nonumber\\
    &\hspace*{-5mm}=\int_{\R^+} \Q\left(\sqrt{\frac{N}{V_e(x)}}\left[\log_2\left(\frac{1+\gamma_b}{1+x}\right) -{R_0}\right]\right)f_{\gamma_e}(x) dx,
\end{align}
where $V_e(x)=\log^2_2\e\frac{x(x+2)}{(x+1)^2}$ and ${R_0}\treq\sqrt{\frac{V_b}{N}}\Qi(\varepsilon)+\frac{m}{N}$. 

It is intractable to obtain a  closed-form expression for \eqref{delta_bar} due to the integration over the complicated composite Gaussian Q-function. In Proposition \ref{prop1}, we derive an approximate expression for  \eqref{delta_bar} using the Laplace’s Approximation Theorem \cite{laplace}. 
\begin{prop}\label{prop1}
The average information leakage (AIL)  $\bar{\delta}$ for the considered wiretap channel can be approximately calculated as
 \begin{align}\label{deltaApprox}
      \bar{\delta} \approx  1-F_{\gamma_e}\left(x_0\right),
 \end{align}
where $F_{\gamma_e}(\cdot)$ is the CDF of $\gamma_e$, and $x_0$ is defined for $x_0\geq 0 $ as 
\begin{align}\label{x0_ineq}
    x_0 = \frac{1+\gamma_b}{2^{R_0}}-1.
\end{align}
 \end{prop}
\begin{proof}
Please see Appendix \ref{AppendixA}
\end{proof}

\begin{remark}\label{remark1}
Note that the AIL approximation derived in this work is general and applicable to any precoding strategy Alice adopts, regardless of the fading channel model. 
\end{remark}

\vspace{-3mm}  \subsection{Secrecy Outage vs. Average Information Leakage}
Recall that a secrecy outage event in the conventional IBL regime occurs when the capacity of the wiretap link $C_e$ is larger than the redundancy rate $R_e$. Accordingly, the secrecy outage probability (SOP) is generally formulated as
\begin{align}\label{pso}
P_{so} &= \Pr\{C_e > R_e\}=1-\Pr\{C_e \leq R_e\},\nonumber \\
& = 1- F_{\gamma_e}\left(2^{R_e}-1\right).
\end{align}
Comparing \eqref{deltaApprox} with \eqref{pso}, it becomes evident that these two expressions are identical for a particular choice of redundancy rate. This provides a compelling link between the well-known SOP in the IBL regime and the AIL in the FBL regime, leading to the following important corollary.
\begin{cor}
The Laplace/saddle-point approximation of the AIL in the FBL regime given in {\normalfont{Proposition \ref{prop1}}} is equivalent to the SOP in the IBL regime given by \eqref{pso} when setting the redundancy rate $R_e = \log_2(1+\gamma_b) -\sqrt{\frac{V_b}{N}}\Qi(\varepsilon)-\frac{m}{N}$.
\end{cor}

\section{Performance Analysis: Case Studies }\label{sec:case}
This section is devoted to performing case studies to characterize AIL under some specific signaling schemes and fading channels. In particular, we carefully analyze the AIL for artificial noise  (AN) beamforming over both Rayleigh and Rician fading channels\footnote{\textcolor{black}{Note that the following analysis can readily be extended to encompass generalized fading models, including Nakagami and log-normal, among others, for specific scenarios. Since different environments may manifest distinct fading behaviors, the chosen model should accurately reflect these conditions for a meaningful analysis with justifiable complexity. In this work, we explore both Rayleigh and Rician fading, two extensively utilized channel models.}}. In doing so, we aim to obtain further insights into the FBL scenario,  verify the accuracy of the approximate results derived in Proposition \ref{prop1} for the AIL, and finally enable the design of key system parameters for secrecy performance improvement.  
 
We assume that Alice adopts an AN beamforming strategy 
for her sporadic secure FBL communication, wherein she transmits the information signal $s$ in conjunction with a $1 \times (k-1)$ AN signal vector $\v$ such that $\v\sim\CN(\mathbf{0},\mathbf{I}_{k-1})$ to confuse Eve and effectively combat eavesdropping \cite{Yan2016}. In particular, letting the power allocation factor $\alpha$, where $0 < \alpha \le 1$, indicate the fraction of the transmit power allocated to the information signals, and assuming an equal power distribution of the AN signal among the transmit antennas, Alice designs the $k\times 1$ transmit signal vector as
\begin{align}\label{x_ani}
    \mathbf{x} =[\w~\U] \begin{bmatrix}
            \sqrt{\alpha}s\\
           \sqrt{\frac{1-\alpha}{k-1}}\v \end{bmatrix}= \sqrt{\alpha} \w s  +\sqrt{\frac{1-\alpha}{k-1}} \U \v ,
\end{align}
where $\W\treq[\w,~\U]$ is the so-called $k\times k$ AN beamforming matrix, in which $\w$ is a $k \times 1$ vector adopted to transmit $s$ while maximizing Bob's instantaneous received SNR, and $\U$ is a $k \times (k-1)$ matrix designed to degrade Eve's channel quality by transmitting $\v$ in all directions except towards Bob. More precisely, Alice computes the eigenvalue decomposition of the matrix $\H=\h^\dagger_b\h_b$, and then selects $\w$ as the principal eigenvector corresponding to the largest eigenvalue of $\H$, and concurrently chooses $\U$ as the remaining eigenvectors of $\H$. Note that the columns of $\U$ are orthonormal and span the nullspace of the main channel $\h_b$, i.e., $\h_b\U=\mathbf{0}$, avoiding any interference to Bob. \textcolor{black}{As such, the received SNR at Bob  can be expressed as
\begin{align}\label{gammab_ani}
    \gamma_b = \alpha \tilde{\gamma}_b,
\end{align} 
where $\tilde{\gamma}_b=\bgamma_b\|\h_b\|^2$ and $\bgamma_b = \rho \beta_b$. As for Eve, the received signal in \eqref{yj} becomes
\begin{align}
y_e= \sqrt{\alpha P \beta_e} \h_e\w s   +\sqrt{\frac{1-\alpha}{k-1} P \beta_e}   \h_e \U \v + \xi_e.   
\end{align}
The corresponding SNR at Eve is then calculated as
\begin{align}\label{gammaE_ani}
    \gamma_e = \alpha \h_e \w \left(\frac{1-\alpha}{k-1} \h_e \U \U^\dagger\h^\dagger_e + \frac{1}{\bgamma_e}\right)^{-1}  \w^\dagger \h^\dagger_e,
\end{align}
where  $\bgamma_e = \rho \beta_e$.}

\textcolor{black}{
\begin{remark}
In our case studies, we consider AN beamforming as the precoding strategy in the MISO wiretap channel. Nevertheless, it is worth pointing out that in a more general MIMO wiretap channel with AN beamforming, when Eve possesses the same or even a larger number of antennas than Alice, significant challenges emerge for ensuring communication confidentiality. With comparable or greater spatial diversity, she can more effectively intercept multiple streams, potentially mitigating the intended interference caused by AN beamforming. This situation complicates interference alignment techniques and necessitates the use of intricate and adaptive resource allocation strategies to maintain an acceptable level of secrecy. Thus, achieving a high secrecy rate or low information leakage, especially in the context of FBL, becomes more challenging. Accordingly, appropriate beamforming and signal processing techniques should be explored for FBL communication systems where an adversary with a comparable or superior number of antennas is present.
\end{remark}}

\vspace{-3mm}  \subsection{Rayleigh Fading}
Here, we obtain the statistical properties of $\gamma_b$ and $\gamma_e$, assuming the communication channels experience Rayleigh fading. For this case, the i.i.d. channel vector $\h_r$ is normally distributed, i.e., $\h_r \sim \CN(\mathbf{0}_{k\times 1}, \I_{k})$.  Thus, the distribution of the main link channel power $\tilde{\gamma}_b$ is expressed as  a sum of squares of $k$ i.i.d. normal r.v.s, and thus follows a  Gamma distribution with shape parameter $k$ and scale parameter $\bgamma_b$, i.e., $\tilde{\gamma}_b \sim \Gam{k}{\bgamma_b}$, whose PDF and CDF are given respectively as \cite{Proakis2006}
\begin{align}\label{cdf_gammaB_Rayl}
    f_{\tilde{\gamma}_b}(x)= \frac{x^{k-1}\e^{-\frac{x}{\bgamma_b}}}{{\bgamma_b}^{k}\Gamma(k)},\quad F_{\tilde{\gamma}_b}(x) = \frac{\gamma(k,\frac{x}{\bgamma_b})}{\Gamma(k)},~x>0,
\end{align}
where $\gamma(s,x)=\int^{x}_{0}t^{s-1}\e^{-t}\diff t$ is the lower incomplete Gamma function and $\Gamma(\cdot)$ is the Gamma function \cite{Gradshteyn2014}. 

To evaluate the statistical characteristics of $\gamma_e$ and the corresponding AIL, we consider two cases:
\subsubsection{\texorpdfstring{$0< \alpha <1~$}~(Case I)}  
With the aid of \cite{Yang2016b}, the PDF of $\gamma_e$ is calculated as
\begin{align}\label{pdf_rayl}
       f_{\gamma_e}(x) = \frac{\tau(x) + \bgamma_e(1-\alpha)}{\alpha \bgamma_e} \e^{-\frac{x}{\alpha \bgamma_e}}\tau(x)^{-k},\quad x\geq 0, 
\end{align}
where $\tau(x) \treq 1+\frac{x(1-\alpha)}{\alpha(k-1)}$. As a result, the approximate AIL for the case of AN beamforming over Rayleigh fading channels is derived by rewriting \eqref{deltaApprox}  as
\begin{align}\label{rayl_ani}
    \bdelta_{Rayl, AN} &\approx 1- \int^{x_0}_0 \frac{\tau(x) + \bgamma_e(1-\alpha)}{\alpha \bgamma_e} \e^{-\frac{x}{\alpha \bgamma_e}}{\tau(x)}^{-k} dx\nonumber\\
    &=   \frac{\e^{-\frac{x_0}{\alpha \bgamma_e}}}{\tau(x_0)^{k-1}}.
\end{align}

\subsubsection{\texorpdfstring{$\alpha=1~$}~(Case II)} When $\alpha=1$, the adopted AN beamforming reduces to maximum ratio transmission (MRT) precoding, allocating all the power for signal transmission without AN. Note that MRT beamforming maximizes the signal-to-noise ratio (SNR) at the intended receiver, Bob, by weighting the transmitted signals based on the channel response of the main link only. As a result, the received SNR at Eve, given by \eqref{gammaE_ani}, can be simplified as
\begin{align}
      \gamma_e &=  \bgamma_e\Big|\h_e\frac{{\h^\dagger_b}}{\|\h_b\|} \Big|^2.
\end{align}

Now, we find the distribution of $\gamma_e$ conditioned on $\h_b$. To this end, we first present the following lemma.
\begin{lemma}\label{lemma1}
    Let $\X\sim \CN(\o_{1 \times N}, \I_N)$ and $\b\in\C^{1\times N}$ with $\|\b\|^2=\sigma^2$. Then, the new r.v. $Y=|\X\b^\dagger|^2$ follows an exponential distribution with mean $\sigma^2$, i.e,  $Y\sim\Exp{\sigma^2}$. 
\end{lemma}
\begin{proof}
   We commence the proof by denoting $\X=\X_R + j \X_I$, where $\X_R=\Re(\X)\sim\N(\o_{1 \times N}, \frac{1}{2}\I_N)$ with $\N(\mu, \sigma^2)$ indicating a real-valued normal distribution with mean $\mu$ and variance $\sigma^2$. Likewise, $\X_I=\Im(\X)\sim\N(\o_{1 \times N}, \frac{1}{2}\I_N)$ and $\b=\b_R+j\b_I$ with $\b_R=\Re(\b)$ and $\b_I=\Im(\b)$. Thus, we can rewrite $Y_1=\X\b^\dagger$ as
    \begin{align}
        Y_1= (\X_R \b^T_R + \X_I \b^T_I) + j (\X_I \b^T_R - \X_R \b^T_I),
    \end{align}
    Then it can be verified that $Y_1\sim\CN(0, \sigma^2)$, hence $Y=Y^2_1 \sim \Exp{\sigma^2}$, which completes the proof.
\end{proof}
According to Lemma \ref{lemma1} and considering the closure of the exponential distribution under scaling by a positive factor,  $\gamma_e$ follows an exponential distribution $\gamma_e \sim \Exp{\bgamma_e}$, whose PDF and CDF are given respectively by
\begin{align}
        f_{\gamma_e}(x)= \frac{1}{\bgamma_e}\e^{{-\frac{x}{\bgamma_e}}},\quad F_{\gamma_e}(x) = 1-\e^{{-\frac{x}{\bgamma_e}}},~x>0 .
\end{align}
As a result, the AIL for the case of Rayleigh fading with $\alpha=1$ (a.k.a. MRT beamforming) is approximately calculated, using \eqref{deltaApprox}, as
\begin{align}\label{rayl_mrt}
   \bdelta_{Rayl, MRT}\approx\exp\left(-\frac{x_0}{\bgamma_e}\right).
\end{align}

\subsection{Rician Fading}
Now, assuming that the communication channels experience Rician fading, we obtain the statistical properties of $\tilde{\gamma}_b$ and $\gamma_e$ as follows. For Rician fading, there exists a line-of-sight (LoS) component together with the non-LoS multipath components. Henceforth, the normalized small-scale fading $\h_r$ can be expressed as the superposition of both the LoS and non-LoS components:
\begin{align}\label{hi}
    \h_r= \sqrt{\frac{K_r}{1+K_r}} \pmb{\zeta}_r+\sqrt{\frac{1}{(1+K_r)}}\tilde{\h}_r,\quad r\in\{b,~e\}
\end{align}
where $\pmb{\zeta}_r$ is composed of unit-modulus complex values, $\tilde{\h}_r\sim \CN(\mathbf{0}_{1\times k},\I_k)$, and $K_r$ is the so-called \textit{K-factor} for the $r$-th link that indicates the power ratio between the LoS component, i.e., $\nu^2_r=\frac{K_r}{1+K_r}$, and scattered non-LoS components, i.e., $\varsigma^2_r=\frac{1}{K_r+1}$. Before proceeding further, we present the following lemma which will be used in the subsequent analysis.

\begin{lemma}\label{lemma2}
Given  $k$ i.i.d. real-valued normal r.v.s with mean $\mu_i$ and unit variance, i.e., $\{X_i\sim\N(\mu_i, 1)\}^n_{i=1}$, the r.v. $Y=\sum^k_{i=1}X^2_i$ obeys a noncentral chi-squared (nc-$\chi^2$) distribution with $k$ degrees of freedom and noncentrality parameter $\lambda=\sum^k_{i=1}\mu^2_i$ such that $Y\sim\chi^2_{nc}(k, \lambda)$. The  PDF of $Y$ is  \cite{Papoulis2002}
\begin{align}
f_{Y}(x) &= \frac{1}{2} \exp\left(-\frac{x+\lambda}{2}\right)\left(\frac{x}{ \lambda}\right)^{\frac{k}{4}-\frac{1}{2}} I_{\frac{k}{2}-1}\left(\sqrt{\lambda x}\right),
\end{align}
where  $I_\nu(\cdot)$ is the modified Bessel function of the first kind with $\nu$-th order given by 
\begin{align}
I_\nu(x) = (\frac{x}{2})^\nu\sum^\infty_{i=1}\frac{(x^2/4)^i}{i!\Gamma(\nu+i+1)}.
\end{align}
Moreover, the CDF of $Y$ is expressed as
\begin{align}
    F_Y(y)&=1-\mathbf{Q}_{\frac{k}{2}}(\sqrt{\lambda},\sqrt{y}),
\end{align}
where $\mathbf{Q}_\nu(\cdot,\cdot)$ denotes the generalized Marcum $\mathbf{Q}$-function of order $\nu>0$, represented for $x,y \geq 0$ as \cite{Gradshteyn2014}
\begin{align}\label{gammaB_rice_cdf}
    \mathbf{Q}_\nu(x,y) =  \int_y^\infty t^{\nu-1}e^{-\frac{1}{2}(x^2+t^2)}I_{\nu-1}(xt)\, dt.
\end{align}
Also, the PDF and the CDF of the r.v. $Z=cY$, with $c\in\mathbb{R}^+$ are
$f_Z(z)= \frac{1}{c}f_Y(\frac{z}{c})$ and $F_{Z}(z)=F_Y(\frac{z}{c})$, respectively. Using moment generating functions  \cite{Papoulis2002}, it can be verified that the sum of nc-$\chi^2$ distributed r.v.s $V=\sum^n_{i=1}Y_i$, where $Y_i\sim \chi^2_{nc}(k_i,\lambda_i)$, also has an nc-$\chi^2$ distribution defined by $V\sim \chi^2_{nc}(k_v,\lambda_v)$ with $k_v=\sum^n_{i=1} k_i$ and $\lambda_v= \sum^n_{i=1} \lambda_i$.
\end{lemma}

Now, we focus on determining the statistical properties of $\tilde{\gamma}_b=\bgamma_b\|\h_b\|^2$. To this end, we  rewrite each entry of $\h_b$ in \eqref{hi} as $h_b =X+ j Y$, where  the real-valued normal r.v.s $X$ and $Y$ are defined as $X\sim \N(\mu_x,\sigma^2)$ and $Y\sim \N(\mu_y,\sigma^2)$, respectively, $\sigma^2=\frac{1}{2(K_b+1)}$  $\mu_x = \sqrt{\frac{K_b}{K_b+1}}\cos(\theta)$, $\mu_y = \sqrt{\frac{K_b}{K_b+1}}\sin(\theta)$, and $\theta$ is an arbitrary real-valued constant. Rewriting $X\sim \sigma\N(\frac{\mu_x}{\sigma},1)$ and $Y\sim \sigma\N(\frac{\mu_y}{\sigma},1)$ and using Lemma \ref{lemma2}, it can be seen that the r.v. $Z=(\frac{X}{\sigma})^2 + (\frac{Y}{\sigma})^2$ follows a nc-$\chi^2$ distribution with $2$ degrees of freedom and noncentrality parameter $\lambda_z=\frac{\mu^2_x+\mu^2_y}{\sigma^2}=2K_b$, i.e., $Z\sim\chi^2_{nc}(2, \lambda_z)$. As a result, $\tilde{\gamma}_b$, expressed as a sum of $k$ i.i.d. nc-$\chi^2$ distributed r.v.s each with $2$ degrees of freedom, has a scaled nc-$\chi^2$ distribution with $2k$ degrees of freedom, noncentrality parameter $2kK_b$, and scale parameter $2\bgamma_b(1+K_b)$. The PDF and the CDF of $\tilde{\gamma}_b$ are derived, using Lemma \ref{lemma2}, respectively as
\begin{align}\label{gammaB_rice_pdf}
f_{\tilde{\gamma}_b}(x) &=  \frac{1}{4\bgamma_b (K_b+1)}\exp\left(-\Big[\frac{x}{4\bgamma_b (K_b+1)} + k K_b\Big]\right)\nonumber\\
&\hspace{-10mm}\times  I_{k-1}\left(\sqrt{\frac{kK_bx}{\bgamma_b(K_b+1)}}\right)\left(\frac{x}{4kK_b\bgamma_b(K_b+1)}\right)^{\frac{k-1}{2}},
\end{align}
and
\begin{align}\label{cdf_gammaB_Rice}
    F_{\tilde{\gamma}_b}(x) &= 1- \mathbf{Q}_{k}\left(\sqrt{2kK_b}, \sqrt{\frac{2(K_b+1)x}{\bgamma_b}} \right).
\end{align}

We now investigate the distribution of $\gamma_e$ for different cases:
\subsubsection{\texorpdfstring{$0 < \alpha <1~$}~(Case I)}
We can rewrite $\gamma_e$ as
\begin{align}
\gamma_e = \frac{{\alpha \bgamma_e}|\h_e \w|^2}{\frac{1-\alpha}{k-1}\bgamma_e \|\h_e \U\|^2 + 1}
&= \frac{X}{Y+1}.
\end{align}
Note that $\h_e$ has i.i.d. complex normal entries and $\W$ is a unitary matrix \cite{Zhou2010a}. Thus, $\h_e\W$ consists of i.i.d. complex normal values, and thus the elements of $\h_e\w$ and $\h_e\U$ are also independent, which further results in the independence of $X=\alpha\bgamma_e|\sum^k_{i=1}h_e(i)w(i)|^2$ and $Y=\frac{1-\alpha}{k-1}\bgamma_e|\sum^k_{p=1}\sum^k_{q=1}h_e(p)\Omega(p,q)h^*_e(q)|^2$, where $\Omega\treq\U\U^\dagger$. In order to find the CDF of $\gamma_e$, considering the independency of $X$ and $Y$, we can write
\begin{align}\label{integ}
    F_{\gamma_e}(z)=\int^\infty_0f_Y(y)\int^{z(y+1)}_0f_X(x) dx dy.
\end{align}
Since evaluating the integral in \eqref{integ} is too difficult owing to the challenge of obtaining exact distributions for $X$ and $Y$, we take an alternative numerical approach similar to \cite{Feng2022}, and approximate $\gamma_e$ as a Gamma distributed r.v.\footnote{
The decision to adopt the Gamma distribution as the approximate model is based on numerous experiments. While more accurate models could be considered, they may entail more parameters and hence higher complexity.}, i.e., $\gamma_e\approx Z \sim\mathcal{G}(a_{z}, b_z)$, where $a_z$ and $b_z$ are shape and scale parameters that can be numerically evaluated through the parameter estimation formulas 
\begin{align}
    a_z = \frac{\V\{\gamma_e\}}{\E\{\gamma_e\}},~ b_z = \frac{\E^2\{\gamma_e\}}{\V\{\gamma_e\}},
\end{align}
where $\V\{x\}=\E\{(x-\E\{x\})^2\}$. Accordingly, the CDF of the r.v. $Z$ is given by
\begin{align}\label{approx_cdf}
    F_Z(z) = \frac{\mathcal{\gamma}(a_z,\frac{z}{b_z})}{\Gamma(a_z)}.
\end{align}

\begin{figure}[t]
    \centering
    \includegraphics[width = \columnwidth]{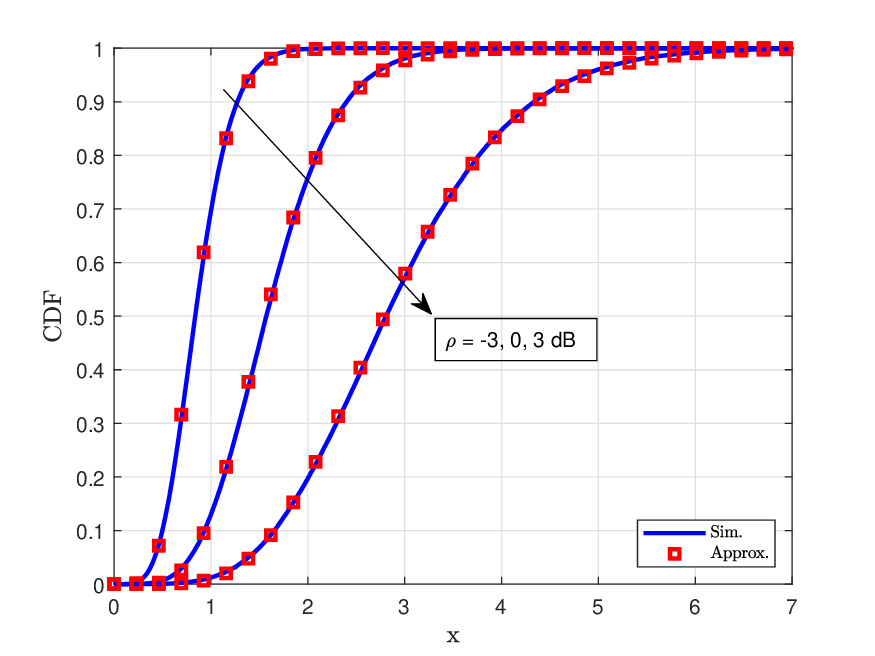}
    \caption{Comparison between  $F_{\gamma_e}(x)$, labeled as {Sim.} and $F_Z(x)$, marked as {Approx.} for different transmit SNRs. The simulation parameters are set as $\alpha = 0.7$, $\beta_b = 3$, $\beta_e = 1$, $K_b=K_e=5$, and $k=4$.}
    \label{simfig1}
\end{figure}
To illustrate the accuracy of the approximate distribution, Fig. \ref{simfig1} shows the results of a Monte-Carlo simulation comparing the true distribution, i.e., $F_{\gamma_e}(x)$ obtained via $100,000$ channel realizations and plotted by solid lines, with the approximate model given by \eqref{approx_cdf}, plotted by marked curves. Fig. \ref{simfig1} illustrates that the approximate model very accurately mirrors the true distribution for different transmit SNRs. As a result, the approximate AIL expression for FBL communication over Rician fading channels  under Case I  is calculated as
\begin{align}\label{rice_ani}
    \bdelta_{Rice, AN}\approx 1- \frac{\mathcal{\gamma}(a_z,\frac{x_0}{b_z})}{\Gamma(a_z)}.
\end{align}

\subsubsection{\texorpdfstring{$\alpha = 1~$}~(Case II)}
For the case of $\alpha=1$, we have $\gamma_e=\frac{\bgamma_e}{\|\h_b\|^2}|\hat{\h}|^2$, where $\hat{h}\treq\h_e \h^\dagger_b$. Consequently, we rewrite $\hat{h}$ using \eqref{hi} as
\begin{align}
\hat{h} &= \sqrt{\frac{K_e}{1+K_e}}\pmb{\zeta}_e\h^\dagger_b +\sqrt{\frac{1}{(1+K_e)}}\tilde{\h}_e \h^\dagger_b,\\
&\stackrel{(a)}{=} X + j Y,
\end{align}
where $(a)$ follows from Lemma \ref{lemma1}, $X\sim\N(\mu_x,\sigma^2_1)$ and $Y\sim\N(\mu_y,\sigma^2_1)$, $\mu_x=\Re(\sqrt{\frac{K_e}{1+K_e}}\pmb{\zeta}_e\h^\dagger_b)$, $\mu_y=\Im(\sqrt{\frac{K_e}{1+K_e}}\pmb{\zeta}_e\h^\dagger_b)$, and $\sigma^2_1=\frac{\|\h_b\|^2}{2(1+K_e)}$. We also note that $\frac{X^2}{\sigma^2_1}\sim \chi^2_{nc}(1, \lambda_x)$, where ${\lambda_x}=\frac{\mu^2_x}{\sigma^2_1}$, $\frac{Y^2}{\sigma^2_1}\sim \chi^2_{nc}(1, \lambda_y)$, and ${\lambda_y}=\frac{\mu^2_y}{\sigma^2_1}$. Using Lemma \ref{lemma2}, $\gamma_e$ thus follows a scaled nc-$\chi^2$ distribution with two degrees of freedom, noncentrality parameter $\lambda_x+\lambda_y=\frac{\mu^2_x+\mu^2_y}{\sigma^2_1} = 2K_e\frac{|\pmb{\zeta}_e\h^\dagger_b|^2}{\|\h_b\|^2}$, and scaling factor $c=\frac{\bgamma_e}{\|\h_b\|^2}\sigma^2_1 = \frac{\bgamma_e}{2(1+K_e)}$. Its CDF can be written  as
\begin{align} 
F_{\gamma_e}(x) &= 1- \mathbf{Q}_{1}\left(\sqrt{2K_e \frac{|\pmb{\zeta}_e\h^\dagger_b|^2}{\|\h_b\|^2}}, \sqrt{\frac{x}{c}}\right),\\
&\stackrel{(b)}{\geq}  1- \mathbf{Q}_{1}\left(\sqrt{2kK_e}, \sqrt{\frac{x}{c}}\right). \label{CDF_rice}
\end{align}
where $(b)$ follows from the Cauchy–Schwarz inequality and considering the fact that $\mathbf{Q}_{1}(x,y)$ is strictly non-decreasing with respect to (w.r.t.) the first argument $x$.
Substituting \eqref{CDF_rice} into the AIL expression in Prop. \ref{prop1}, we  obtain an approximate upper-bound expression for the AIL in FBL communication over Rician fading  under Case II  as
\begin{align}\label{rice_mrt}
    \bdelta_{Rice, MRT}\approx  \mathbf{Q}_{1}\left( \sqrt{2kK_e}, \sqrt{\frac{2(1+K_e)x_0}{\bgamma_e}}\right).
\end{align}

\vspace{-3mm}  \subsection{High-SNR Analysis of AIL}\label{section:highSNR}
In this section, we analyze the performance of AIL in the high-SNR regime to obtain more insights into system performance and the key design parameters. Thus, we evaluate $\bdelta$ when $\rho$ approaches infinity for all of our considered scenarios. Prior to this, we examine the asymptotic behavior of $x_0$ at high SNRs, denoted by $x^{\infty}_0$, as
\begin{align}
     x^{\infty}_0 =\lim_{\rho\to\infty} x_0 = \frac{\gamma_b}{\exp({R}^\infty_0\ln 2)},
\end{align}
where ${R}^\infty_0$ is a positive constant represented by
\begin{align}\label{Rinf}
    {R}^\infty_0 = \frac{\log_2(\e)}{\sqrt{N}}\Qi(\varepsilon)+\frac{m}{N}.
\end{align}

\subsubsection{Rayleigh Fading}
We compute the high-SNR approximation of AIL for Case I, using \eqref{rayl_ani}, as
\begin{align}\label{rayl_caseI_inf}
     \bdelta^{\infty}_{Rayl, AN}&=\lim_{\rho\to\infty} \bdelta_{Rayl, AN}\\
     &\rightarrow \frac{\e^{-\frac{\gamma_b}{\alpha \bgamma_e 2^{{R}^\infty_0}}}}{(\frac{1-\alpha}{\alpha(k-1)}x^\infty_0)^{k-1}}.
\end{align}
Similarly, the high-SNR approximation of AIL for Case II given by \eqref{rayl_mrt} is calculated as
\begin{align}
    \bdelta^\infty_{Rayl, MRT}&=\lim_{\rho\to\infty} \bdelta_{Rayl, MRT} \\
    &\rightarrow {\e^{-\frac{\gamma_b}{\bgamma_e {R}^\infty_0}}}.
\end{align}

\subsubsection{Rician Fading}
For Rician fading scenarios, the high-SNR assumption for Case I and Case  II leads respectively to
\begin{align}
    \bdelta^\infty_{Rice, AN}&=\lim_{\rho\to\infty} \bdelta_{Rice, AN} \nonumber\\
    &\rightarrow 1- \frac{\mathcal{\gamma}(a_z,\frac{\gamma_b}{{R}^\infty_0 b_z})}{\Gamma(a_z)}
\end{align}
and 
\begin{align}
    \bdelta^\infty_{Rice, MRT}&=\lim_{\rho\to\infty} \bdelta_{Rice, MRT} \\
    &\rightarrow\mathbf{Q}_{1}\left( \sqrt{2kK_e}, \sqrt{\frac{2(1+K_e)\gamma_b}{2^{{R_0}^{\infty}}\bgamma_e}}\right).
\end{align}

\textcolor{black}{Based on the above results, it is clear that the AN beamforming strategy can achieve an AIL performance that asymptotically approaches zero, irrespective of the channel fading type, i.e., $\bdelta^\infty_{Rice, AN}=\bdelta^\infty_{Rayl, AN}\rightarrow 0$ when $\rho$ goes to infinity.  This stands in stark contrast to the MRT beamforming strategy, where the associated AIL leads to a non-zero value regardless of how large the transmit SNR becomes. In other words, when there are sufficient resources, the AN beamforming can be effectively designed  such that the AIL to Eve can be reduced to any arbitrarily low level.}

\begin{remark}  Despite providing valuable insight into the system performance, we will see later that the high-SNR approximation, which leads to a constant channel dispersion, is only accurate when the transmit SNR is very high. Consequently, in order to avoid oversimplifying our conclusions, we need to carefully treat the impact of the SNR within the region of interest when designing the system.
\end{remark}

\section{Secure Transmission Design}\label{sec:problem}
\textcolor{black}{In this section, we turn our attention to the problem of optimizing key system parameters such as the blocklength $N$ for a fixed number of transmit information bits $m$ and power allocation factor $\alpha$ in the considered secure FBL communication system.} Using the definition of AIL explored in the previous section for different fading scenarios and beamforming schemes,  we define the instantaneous secrecy throughput (IST) for a given $\tilde{\gamma}_b[i]$ in bits per channel use (bpcu) as 
\begin{align}
    \mathcal{T}[i]&=[1-\varepsilon] \frac{m}{N[i]} \mathbbm{1}_{\bdelta[i] \leq \phi},
\end{align}
where $i\in\{1, 2, \cdots, L\}$ indicates the timeslot index of the FBL transmissions for different block-fading channels, and the indicator function $\mathbbm{1}_C$ is defined as
\begin{align}
    \mathbbm{1}_C = \begin{cases}
       1,& \text{if condition C is true}\\
        0,& \text{otherwise}
    \end{cases}
\end{align}
The critieron $\mathbbm{1}_{\bdelta[i] \leq \phi}$ ensures that the AIL does not exceed the predetermined threshold $\phi$, which represents the minimum acceptable secrecy level. We formulate our design problem  in terms of maximizing the average secrecy throughput (AST) under the given reliability and security requirements as
\begin{subequations}
\begin{align}\label{opt_prob_original}
\Prob{}:&~ \stackrel{}{\underset{\{\mathbf{N}, \pmb{\alpha}\}}{\mathrm{max}}}~~ \bar{\mathcal{T}} = \frac{1}{L}\sum^L_{i=1} \mathcal{T}[i]
\nonumber\\
&~~\text{s.t.}~~\cst{1}:~ N[i] \leq N^{max},~N[i]\in\mathbb{Z}^{+},~\forall i\\
&~\qquad~\cst{2}:~0 < \alpha[i] \leq 1,~\forall i
\end{align}
\end{subequations}
where \textcolor{black}{$\mathbf{N}=\{N[1], N[2], \cdots, N[L]\}$ and $\pmb{\alpha} = \{\alpha[1], \alpha[2], \cdots, \alpha[L]\}$ indicate optimization block variables},  \cst{1} and \cst{2} impose practical bounds on the design parameters, $N^{max}$ is the blocklength corresponding to the maximum tolerable delay of the FBL transmission per timeslot\footnote{Note that total transmission delay captures other latency parameters such as propagation delay and the encoding/decoding processing delay at the transceivers. For the sake of simplicity, we assume that all these delays are reflected in the maximum blocklength parameter.}. We emphasize that \Prob{} is generally a challenging stochastic optimization problem with a complicated objective function and mixed-integer variables. In what follows, we delve into tackling \Prob{} via both adaptive and non-adaptive designs.

\vspace{-3mm}  \subsection{Adaptive Design}
In the adaptive design, since Alice possesses complete knowledge of Bob's instantaneous CSI, we assume that she has the ability to adapt the parameter vectors $\mathbf{N}$ and $\pmb{\alpha}$ according to the instantaneous condition of the main channel $\tilde{\gamma}_b[i],~\forall i$. We equivalently divide \Prob{} into $L$ subproblems due to the i.i.d. channel assumption and solve them separately. The formulation for the $i$-th subproblem is expressed as
\begin{subequations}
\begin{align}
\Prob{1.i}:&~ \stackrel{}{\underset{\{N[i], \alpha[i]\}}{\mathrm{max}}}~~\frac{m(1-\varepsilon)}{N[i]}
\nonumber\\
&~~\text{s.t.}~~N[i] \leq N^{max},~N[i]\in\mathbb{Z}^{+},~\forall i\label{58a}\\
&~\qquad~0 < \alpha[i] \leq 1,~\forall i\\
&~\qquad~\bdelta[i] \leq \phi,~\forall i . \label{58c}
\end{align}
\end{subequations}
Note that  \Prob{1.i} is a non-convex mixed-integer optimization with non-convex constraints \eqref{58a} and \eqref{58c}, and can be solved by any heuristic search method. For example, we can commence the optimization by setting $N[i]=1$ owing to the fact that the objective function is decreasing w.r.t. $N[i]$, and then perform a one-dimensional (1D) search on the domain of $\alpha[i]$ until the constraint \eqref{58c}  holds.  If the constraint cannot be satisfied, we then sequentially increment $N[i]$ up to $N^{max}$, making the inequality constraints looser, and conduct the 1D search. If the inequality does not hold for any $N[i]$ and $\alpha[i]$, it implies that the problem is not feasible and the objective value is set to zero. Note that the collated optimal solutions to the $L$ subproblems \Prob{1.i} are equivalent to the solution to \Prob{1}. \textcolor{black}{This approach leads to a worst-case time complexity of $\mathcal{O}(LN^{max}S)$, where $S$ is the size of the search space.}

We also consider an alternative low-complexity approach by dividing \Prob{1} into two subproblems and solving them alternately, which enables us to obtain further insight into the design and reduce the complexity. In particular, for a given fixed feasible $\pmb{\alpha}$ and relaxing the integer optimization variables $\mathbf{N}$ into a continuous-valued quantity $\tilde{\mathbf{N}}$, we recast \Prob{1} approximately as 
\begin{subequations}
\begin{align}
\Prob{2.1}:&~ \stackrel{}{\underset{\tilde{\mathbf{N}}}{\mathrm{max}}}~~\frac{1}{L}\sum^L_{i=1}\frac{m(1-\varepsilon)}{\tilde{N}[i]}
\\
&~~\text{s.t.}~~\vartheta (\phi) \leq \tilde{N}[i] \leq N^{max},~\forall i \label{59a}
\end{align}
\end{subequations}
where $\vartheta(\phi;\alpha[i],\tilde{\gamma}_b[i])$ evaluates the inverse of $\bdelta(\tilde{N}[i];\alpha[i],\tilde{\gamma}_b[i])$ w.r.t. the parameter $\tilde{N}[i]$  such that $\bdelta(\vartheta_i(\phi;\alpha[i],\tilde{\gamma}_b[i]);\alpha[i],\tilde{\gamma}_b[i]) = \phi$. We note that the lower-bound in \eqref{59a} follows from the approximate AIL $\bdelta$ in Proposition \ref{prop1} and from the fact that $\bdelta$ is a decreasing function w.r.t. the blocklength $N$ since the CDF is monotonically increasing. Thus, the optimal solution to \Prob{2.1} exists only if  $\vartheta(\phi;\alpha[i],\tilde{\gamma}_b[i]) \leq N^{max},~\forall i$. 

\textcolor{black}{Due to the monotonically decreasing objective function, the optimal solution for the relaxed blocklength vector $\mathbf{\tilde{N}}^*=\{\tilde{N}^*[1], \tilde{N}^*[2], \cdots, \tilde{N}^*[L]\}$ is given by
\begin{align}
    \tilde{N}^*[i] = \vartheta(\phi;\alpha[i],\tilde{\gamma}_b[i]),~\forall i .
\end{align}}
Since  $\mathbf{\tilde{N}}^*$ can be considered to be a function of $\alpha[i]$ with the given $\tilde{\gamma}_b[i]$, i.e.,  $\vartheta(\alpha[i]; \phi,\tilde{\gamma}_b[i])$, we plug it into the objective function  of \Prob{} to obtain the optimal solution for $\pmb{\alpha}$, denoted by $\pmb{\alpha}^*$, and solve the following optimization problem:
\begin{subequations}
\begin{align}
\Prob{2.2}:&~ \stackrel{}{\underset{{\pmb{\alpha}}}{\mathrm{max}}}~~\psi({\pmb{\alpha}}) = \frac{1}{L}\sum^L_{i=1}\frac{m(1-\varepsilon)}{\vartheta(\alpha[i]; \phi,\tilde{\gamma}_b[i])}
\\
&~~\text{s.t.}~~0  < \alpha[i] \leq 1,~\forall i\textcolor{black}{.}
\end{align}
\end{subequations}
Note that \Prob{2.2} can be solved through a simple bisection search for 
$\nabla_{{\pmb{\alpha}}}\psi({\pmb{\alpha}}^*)= \frac{\partial \vartheta(\alpha[i]; \phi,\tilde{\gamma}_b[i])}{\partial \alpha[i]}\Big|_{\alpha[i]=\alpha^*[i]}=0~\forall i $ in the range $(0,1]$. \textcolor{black}{Once ${{\alpha}}^*[i]$ is found as a function of $\tilde{\gamma}_b[i]$ and the parameter $\phi$, denoted as ${{\alpha}}^*[i]=\tau_1(\tilde{\gamma}_b[i];\phi)$, the optimal solution for the integer-valued blocklength can be calculated as $\mathbf{N}^*=\lceil\vartheta(\tau_1(\tilde{\gamma}_b[i];\phi), \phi,\tilde{\gamma}_b[i])\rceil\treq\upsilon(\tilde{\gamma}_b[i];\phi)$. Note that both $\alpha^*[i]$ and $N^*[i]$ are functions of $\tilde{\gamma}_b[i],~\forall i$. Thus, the optimal AST is obtained as
\begin{align}\label{opt_adapt_T}
    \bar{\mathcal{T}}^* &= \frac{1}{L}\sum^L_{i=1} \frac{m(1-\varepsilon)}{\upsilon(\tilde{\gamma}_b[i];\phi)}.
\end{align}
When $L$ is sufficiently large, \eqref{opt_adapt_T} can be approximated as
\begin{align}\label{approxTbar}
    \bar{\mathcal{T}}^* \approx m(1-\varepsilon) \int^\infty_{\gamma_0} \frac{f_{\tilde{\gamma}_b}(x)}{\upsilon(x;\phi)} \diff x,
\end{align}
where $\gamma_0$ is the nonnegative solution to 
\begin{align}\label{gamma0_eq}
    \ln(1+\tau_1(x,\phi)x)= \sqrt{\frac{\big(1-\frac{1}{{(\tau_1(x,\phi)x+1)}^2}\big)}{\upsilon(x;\phi)}}\Qi(\varepsilon)+\frac{m \ln 2}{\upsilon(x;\phi)}.
\end{align}}
We stress that the approximation $\eqref{approxTbar}$ follows from the law of large numbers \cite{Papoulis2002} and the i.i.d. assumption of r.v.s $\tilde{\gamma}_b[i]~\forall i$. \textcolor{black}{Furthermore, the positive lower limit of the integration in \eqref{approxTbar} is to ensure that the integral is calculated over the feasible range of $\tilde{\gamma}_b$ such that
the constraint \eqref{58c} in the original subproblem \Prob{1.i} (i.e., $\bar{\delta} \leq \phi$) holds while guaranteeing validity of $\bar{\delta}$ in terms of $x_0$, a non-negative parameter of the approximate AIL in Proposition \ref{prop1} given by \eqref{x0_ineq}.
Otherwise, if $\tilde{\gamma}_b$ is too small, the constraint \eqref{58c} may not be satisfied or $x_0$ may become negative, rendering $\tilde{\gamma}_b \geq 0.5$, which is practically invalid. For $\tilde{\gamma}_b$ larger than the threshold $\gamma_0$, the required secrecy constraint \eqref{58c} is satisfied for the obtained ($N^*$, $\alpha^*$), and $x_0$ is positive. Hence, \eqref{gamma0_eq} indeed finds the minimum value of $x= \tilde{\gamma}_b$ over $\R^+$ characterized by $\gamma_0$, via finding the solution to $x_0(\alpha^*, N^*, x)\geq 0$  with equality.}

\vspace{-3mm}  \subsection{Non-Adaptive Design}
To reduce the complexity of improving the AST, Alice can take a non-adaptive design approach in which she adopts the same blocklength  $N$ and power allocation factor $\alpha$ for all the FBL transmissions based on statistical knowledge of $\tilde{\gamma}_b$.  In light of this, the AST optimization can be approximately expressed as
\begin{subequations}
\begin{align}\label{opt_prob_nonadaptive}
\Prob{3}:&~ \stackrel{}{\underset{\{N, \alpha\}}{\mathrm{max}}}~~\bar{\mathcal{T}}\approx\frac{m(1-\varepsilon)}{N}\E\{\mathbbm{1}_{\bdelta \leq \phi}\}
\\
&~~\text{s.t.}~\cst{1}:~  N\leq N^{max},~N\in\mathbb{Z}^{+}\\
&\qquad~\cst{2}:~0 < \alpha \leq 1,
\end{align}
\end{subequations}
where the approximate objective for \Prob{3} is obtained by applying the law of large numbers. 

\textcolor{black}{Note that the objective of \Prob{3} is not in closed form due to the expectation operator. 
To tackle this issue, we obtain a tractable formulation of the objective function. For ease of exposition, we replace ${R}_0$ with ${R}^\infty_0$, given in \eqref{Rinf}, in the formulation of $\bar{\delta}$ in Proposition \ref{prop1}, and conduct some algebraic manipulations to approximately reformulate \Prob{3} as
\begin{subequations}
\begin{align}\label{opt_prob_nonadaptive_1}
\Prob{4}:&~ \stackrel{}{\underset{\{N, \alpha\}}{\mathrm{max}}}~~\frac{m(1-\varepsilon)}{N}\Pr(\tilde{\gamma}_b \geq \gamma_1)
\\
&~~\text{s.t.}~\cst{1}:~  N\leq N^{max},~N\in\mathbb{Z}^{+}\\
&\qquad~\cst{2}:~0 < \alpha \leq 1,
\end{align}
\end{subequations}
where $\gamma_1$ is given by 
\begin{align}\label{on-off}
    \gamma_1= \frac{2^{\frac{\log_2(\e)}{\sqrt{N}}\Qi(\varepsilon)+\frac{m}{N}}(1+\Omega(1-\phi;\alpha))-1}{\alpha}
\end{align}
and $\Omega(x;\alpha)$ is the inverse CDF or quantile function of $F_{\gamma_e}(x)$, i.e., $F_{\gamma_e}(\Omega(x;\alpha))=x$. Since ${R}^\infty_0 > {R_0}$ and $\bdelta$ is an increasing function of ${R_0}$, we can verify that the objective value of \Prob{4} is a lower bound to that of \Prob{3}.}
We also stress that the first multiplicative term $\frac{m(1-\varepsilon)}{N}$ in the objective of \Prob{4} is non-increasing w.r.t. $N$, while the second term $\Pr(\tilde{\gamma}_b \geq \gamma_1)$ is non-decreasing. This implies an inherent trade-off in terms of the impact of blocklength on the AST, which reinforces the significance of the appropriate blocklength design. 

Since the probability in the objective function of \Prob{4} can be written analytically as $\Pr(\tilde{\gamma}_b \geq \gamma_1)  = 1-F_{\tilde{\gamma}_b}(\gamma_1)$, the analytical expression for $\bar{\mathcal{T}}$  over Rayleigh and Rician fading channels using \eqref{cdf_gammaB_Rayl} and \eqref{cdf_gammaB_Rice} are computed as
\begin{align}
   \bar{\mathcal{T}}_{Rayl} \approx \frac{m(1-\varepsilon)}{N} \left[ 1 -  \frac{\gamma(k,\frac{\gamma_1}{\bgamma_b})}{\Gamma(k)} \right],
   \end{align}
and
\begin{align}
\bar{\mathcal{T}}_{Rice} \approx \frac{m(1-\varepsilon)}{N} \mathbf{Q}_{k}\left(\sqrt{2kK_b}, \sqrt{\frac{2(K_b+1)\gamma_1}{\bgamma_b}} \right),
\end{align}
respectively. Although \Prob{4} has much less complexity than \Prob{3}, since it has only two optimization variables, obtaining a closed-form solution is still challenging due to the mixed-integer nonlinear nature of the optimization problem and the
complicated objective function. However, we can numerically solve \Prob{4} with any nonconvex optimization tool such as \cite{OptimizationToolbox}.

\section{Numerical Results and Discussion}\label{sec:simulations}
In this section, we provide simulation results to first validate the accuracy of our analytical expressions obtained for the AIL in Section \ref{sec:case} and investigate the impact of key system parameters such as blocklength, channel condition, transmit power, number of antennas, decoding error probability, and power allocation factor. Then, we present numerical results on AST enhancement by carefully designing the key system parameters, including the coding blocklength and the power allocation factor. Unless otherwise stated, the simulation parameters are set as follows: Number of transmit bits $m=100$, transmit SNR $\rho=0$ dB, decoding error probability $\varepsilon = 10^{-3}$, maximum AIL threshold $\phi=10^{-4}$, \textit{K-factor} for both the main and wiretap links $K_b=K_e \treq K=5$, number of transmit antennas $k=4$, transmit blocklength $N=400$, maximum blocklength $N^{max}=1000$, number of FBL transmissions $L=1000$, $\beta_b \|\h_b\|^2= 3$, and $\beta_e= 1$.

\vspace{-3mm}  \subsection{Verification of Theoretical Results}
\begin{figure*}
\centering
\begin{subfigure}{.5\textwidth}
  \centering
  \includegraphics[width=\linewidth]{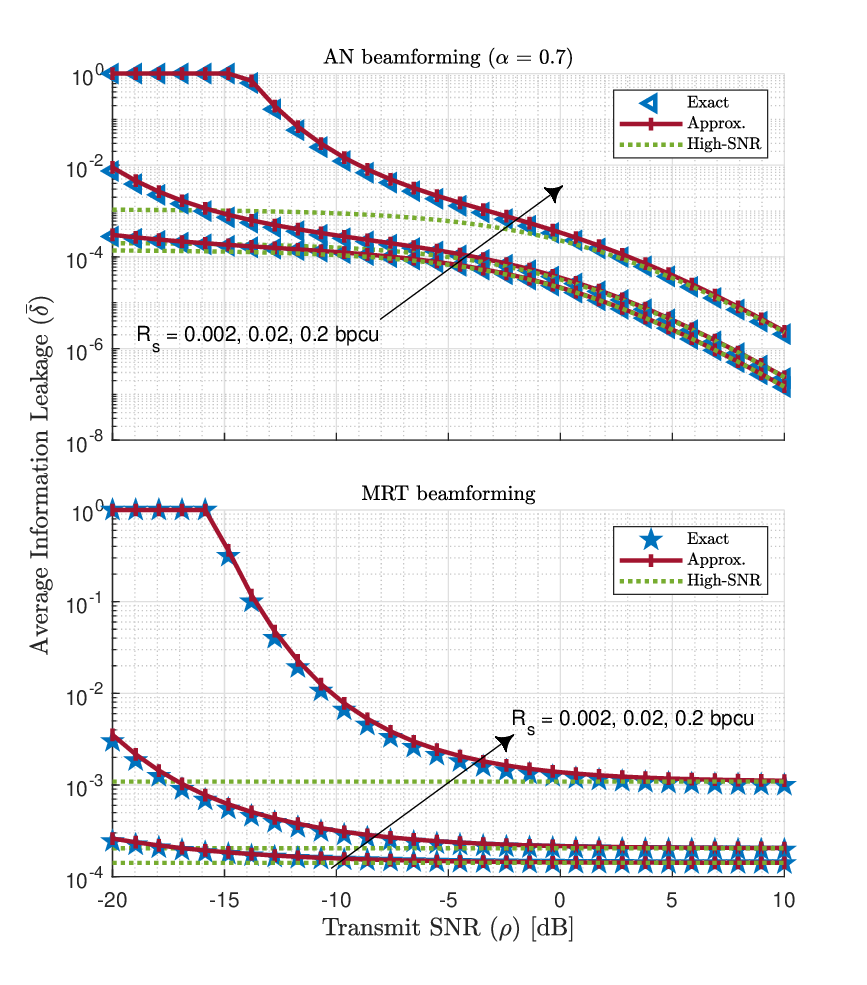}
  \caption{Rayleigh fading.}
  \label{sim1:rayl}
\end{subfigure}%
\begin{subfigure}{.5\textwidth}
  \centering
  \includegraphics[width=\linewidth]{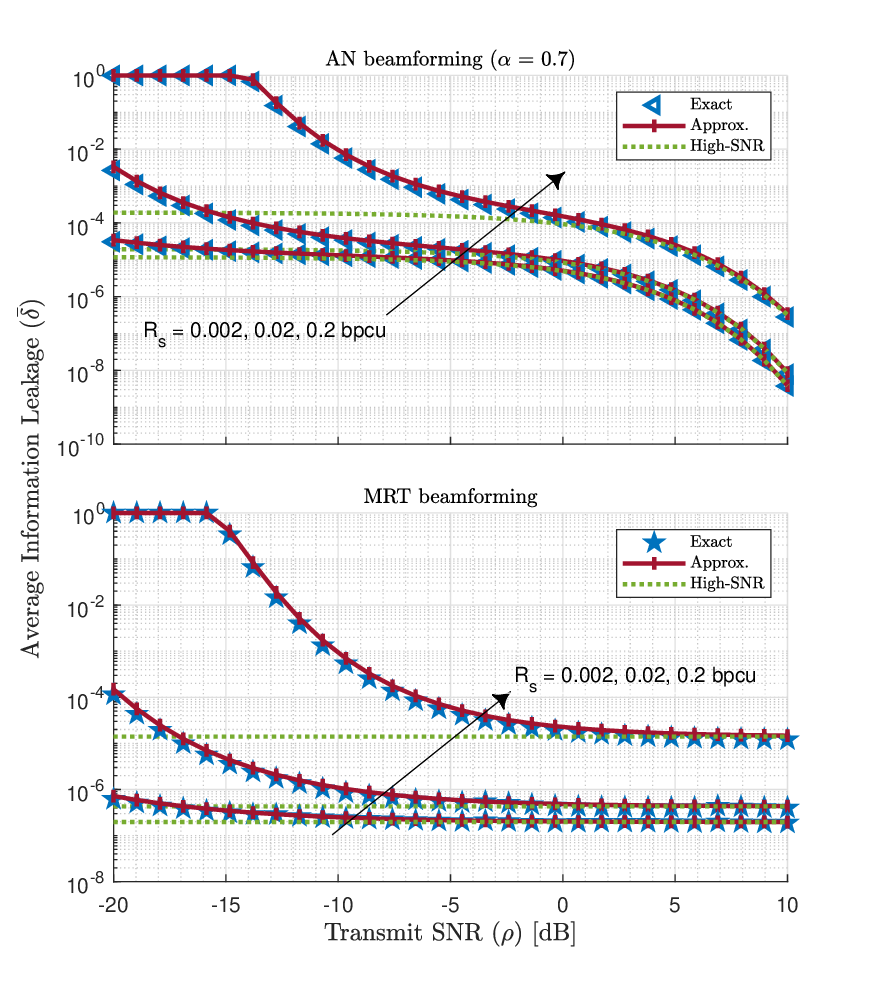}
  \caption{Rician fading.}
  \label{sim1:rice}
\end{subfigure}
\caption{AIL vs. transmit SNR for different fading channels and beamforming schemes. }
\label{sim1}
\end{figure*}
Fig. \ref{sim1} shows the AIL performance for both Rayleigh and Rician fading scenarios with AN and MRT beamforming strategies for a fixed main channel realization and different secrecy coding rates ranging from $R_s=0.002$ bpcu (relatively low) to $R_s=0.2$ bpcu (relatively high). The markers labeled by \textit{Exact} represent the exact AIL evaluation according to \eqref{delta_bar}. The curves labeled as \textit{Approx.} indicate the analytical approximate AIL in the FBL regime according to Proposition \ref{prop1}, and those labeled as \textit{High-SNR} depict the high-SNR evaluations of the AIL analyzed in Section \ref{section:highSNR}. We can see that the theoretical and approximate AILs are in good agreement for a wide range of transmit SNRs, validating our analysis.  Furthermore, for both Rayleigh and Rician fading, our theoretical high-SNR predictions provide accurate asymptotes for the performance for SNRs of 0dB or higher. 
For a given transmit SNR, we observe that the higher the secrecy rate $R_s$ (or equivalently the smaller blocklength since $m$ is fixed), the larger the AIL, regardless of the fading or beamforming scenarios. On the other hand, when the blocklength is fixed, one can achieve stringent AIL by increasing the SNR for AN beamforming. However, with MRT beamforming, we may not achieve the required AIL by solely increasing the transmit SNR, since the AIL approaches a non-zero value when $\rho$ becomes sufficiently large.
Overall, these figures confirm the accuracy of our approximate expressions derived for the AIL performance in Section \ref{sec:performance}, and once again reveal the inherent close link between the  AIL in the FBL regime and the widely known SOP for the IBL case.

\begin{figure}[t]
\centering
\includegraphics[width=\columnwidth]{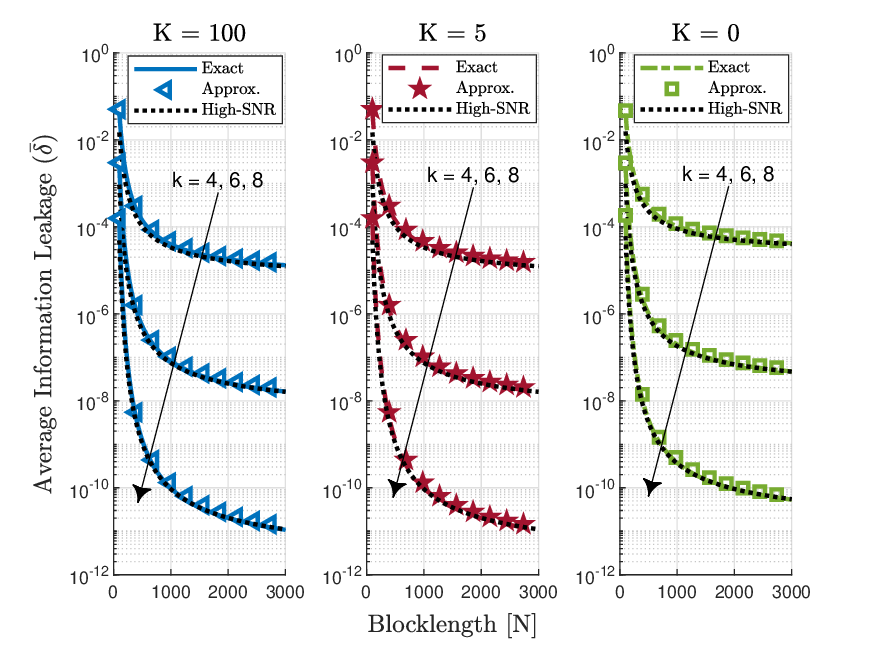}
\caption{AIL vs. blocklength for Rician fading with AN beamforming ($\alpha = 0.7$) for different K-factors and numbers of Alice's antennas. }
\label{sim2}
\end{figure}
In Fig. \ref{sim2}, we investigate the impact of the Rician \textit{K-factor} and the number of Alice's antennas on the AIL for Rayleigh and Rician fading with AN beamforming. Note that Rayleigh fading corresponds to the case of $K=0$, while large values of $K$ (e.g., $K \ge 20$) approximately correspond to a non-fading LoS channel. It is evident from the figures that the derived approximate and high-SNR expressions for the AIL are accurate for a broad range of channels. We see that the LoS-dominant channel ($K=20$ dB) has approximately identical AIL compared to Rician fading with $K=5$ for different values of blocklength, and both yield relatively smaller AIL values compared to Rayleigh fading for the same coding blocklength. In addition, we observe that the larger the number of antennas, the smaller the AIL, and one can obtain a reduction of $10^{-3}$ in AIL by adding only two extra transmit antennas. This demonstrates the significance of the number of antennas and the effectiveness of AN beamforming in boosting the secrecy performance of the FBL transmissions.

\begin{figure}[t]
\centering
\includegraphics[width=\columnwidth]{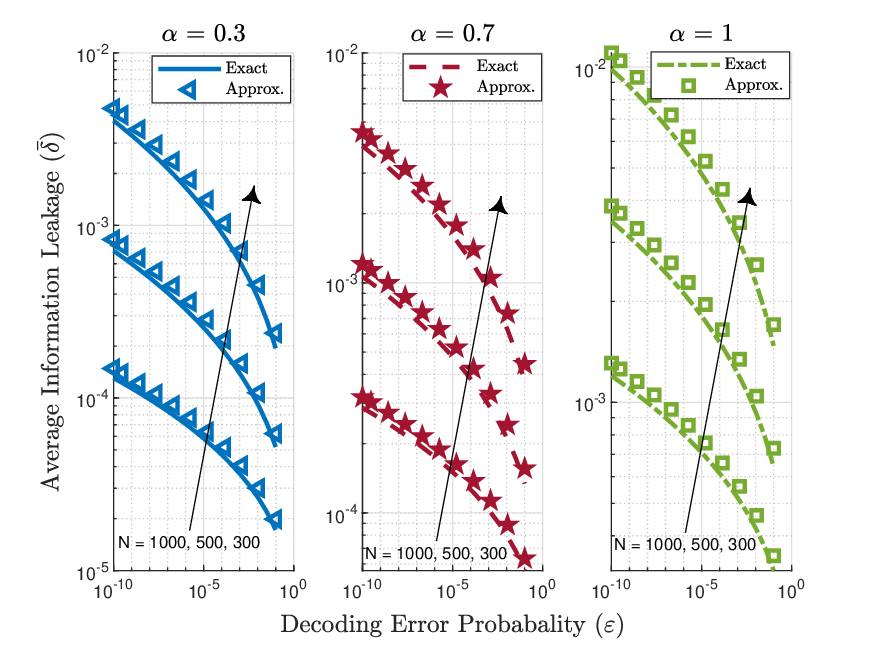}
\caption{AIL vs. decoding error probability for Rayleigh fading.}
\label{sim3}
\end{figure}

\begin{figure}[t]
\centering
\includegraphics[width=\columnwidth]{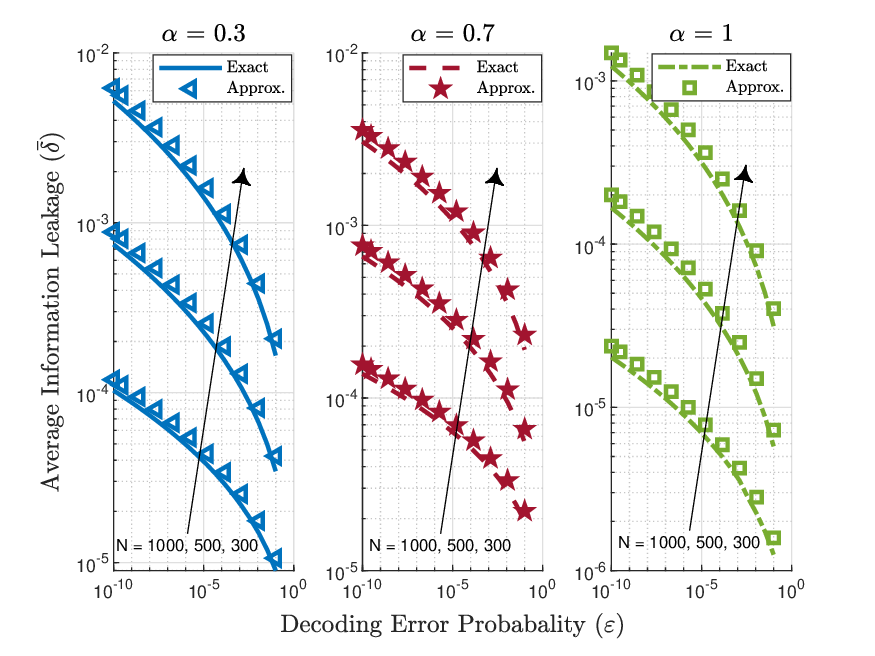}
\caption{AIL vs. decoding error probability for Rician fading.}
\label{sim4}
\end{figure}
Figs. \ref{sim3} and \ref{sim4} provide insights into the security-reliability trade-off for FBL transmissions under Rayleigh and Rician fading, respectively. In particular, we plot the AIL against the decoding error probability $\varepsilon$ for different power allocations and blocklengths. As can be clearly seen from the figures, when $\varepsilon$ increases, the AIL becomes smaller, and this trend becomes more significant as the coding blocklength is reduced. We can also conclude from the figures that a slight compromise in AIL leads to a large reliability improvement. For example, for the case of Rayleigh fading in Fig. \ref{sim3} with $\alpha=0.3$ and $N=300$, if the secrecy level is reduced from $\bdelta = 2\times10^{-4}$ to $\bdelta =10^{-3}$, a loss of $7$ dB, the communication reliability can be improved by 30dB, from $\varepsilon= 10^{-1}$ to $\varepsilon=10^{-4}$. \textcolor{black}{Furthermore, comparing Figs. \ref{sim3} and \ref{sim4}, we observe that Rician fading channels achieve comparably lower AIL than Rayleigh fading channels. In addition, as the signal power ratio $\alpha$ increases, the benefits of the LoS component in Rician fading scenarios become more pronounced in terms of helping reduce the AIL for a given decoding error probability, compared with Rayleigh fading. For example, with $N=500$ and $\alpha=1$, Rician fading leads to the AIL $\bar{\delta}=4\times10^{-5}$, while Rayleigh fading exhibits significantly larger AIL at $\bar{\delta}= 2 \times 10^{-3}$, which corresponds to an approximately 50-fold increase for the given decoding error probability $\varepsilon=10^{-5}$. }

\vspace{-3mm}  \subsection{Adaptive and Non-Adaptive Designs}
\begin{figure}[t]
\centering
\includegraphics[width=\columnwidth]{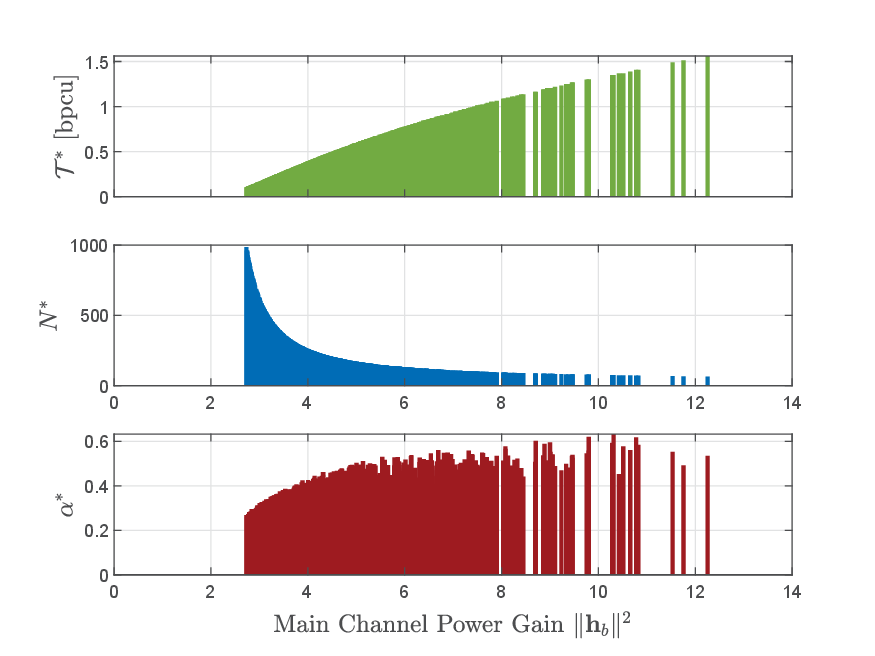}
\caption{Optimal IST, blocklength, and power allocation vs. the main channel power gain $\|\h_b\|^2$ under Rayleigh fading with AN beamforming via the adaptive design.}
\label{sim5}
\end{figure}
Fig. \ref{sim5} depicts the optimal IST for the adaptive design, as well as the optimal blocklength $N^*$ and power allocation $\alpha^*$ for $L=1000$ realizations of the main link $\|\h_b\|^2$ sorted in ascending order. It is clear that as the main channel quality improves, a smaller blocklength is required by the adaptive algorithm to significantly boost the IST, leading to an overall AST performance enhancement.  In addition, Fig. \ref{sim5} shows that a smaller portion of FBL transmit power should generally be dedicated to the confidential signals corresponding to smaller $\alpha$ when Bob experiences relatively low channel quality, e.g., between $\|\h_b\|^2=2.7$ and $\|\h_b\|^2=5.5$. Such a power allocation is required to confuse Eve with stronger AN, while satisfying the required AIL secrecy level.

\begin{figure}[t]
\centering
\includegraphics[width=\columnwidth]{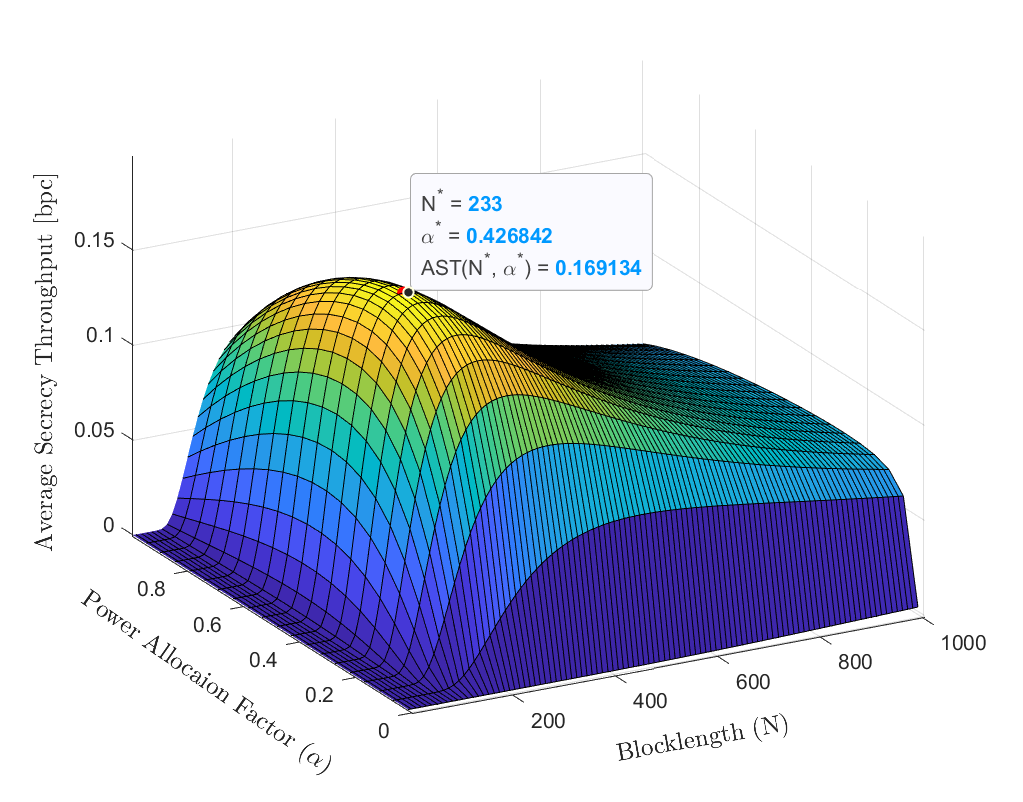}
\caption{AST vs. $\alpha$ and $N$ for the non-adaptive design in Rayleigh fading.}
\label{sim6}
\end{figure}
Fig. \ref{sim6} illustrates the performance of AST for the non-adaptive design under Rayleigh fading with AN beamforming for different blocklengths and power allocations. We see that the AST is jointly non-convex w.r.t. $N$ and $\alpha$ and peaks at $(N^* = 233, \alpha^*=0.43)$ for the given system setting. This can be explained as follows: When $N$ is small, the secrecy requirement plays an important role in significantly decreasing the AST. By increasing $N$ up to a certain threshold, the AST is improved due to the favorable trade-off between $N$ and $\Pr(\tilde{\gamma}_b \geq \gamma_1)$ in the AST given in \Prob{4}. When $N$ exceeds this threshold, the multiplicative term $\Pr(\tilde{\gamma}_b \geq \gamma_1)$ becomes less significant in determining the overall AST performance. Overall, this figure highlights the importance of proper blocklength design and power allocation.

\begin{figure}[t]
\centering
\includegraphics[width=\columnwidth]{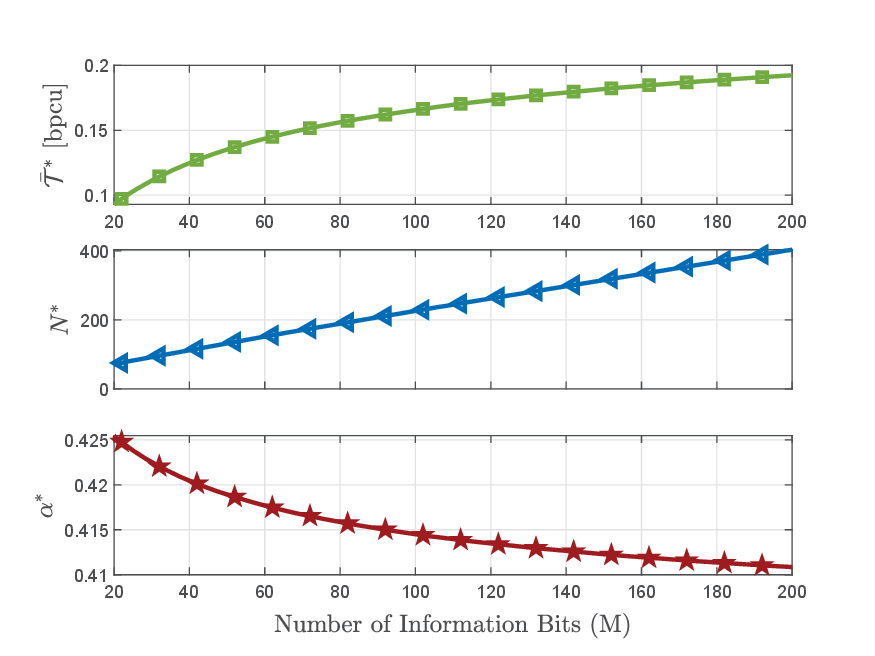}
\caption{Non-adaptive AST design under Rayleigh fading with AN beamforming against different numbers of confidential information bits ($m$). \textcolor{black}{The decoding error probability is set as $\varepsilon
= 10^{-3}$.}}
\label{sim7}
\end{figure}
Fig. \ref{sim7} exhibits the impact of the number of transmit information bits ($m$) on the designed parameters $N^*$ and $\alpha^*$  as well as the optimal AST $\bar{\mathcal{T}}^*$ using the non-adaptive strategy for Rayleigh fading. We observe that a linear increase in $m$ yields a linear increase in $N^*$, while $\alpha^*$ monotonically decreases to improve the AST and maintain the required secrecy level. Furthermore, the overall performance of $\bar{\mathcal{T}}^*$ is non-decreasing w.r.t. $m$; however, due to the maximum blocklength threshold $N^{max}$, a ceiling on the optimal AST is expected to appear when $m$ becomes sufficiently large.

\section{Conclusion and Future works}\label{sec:conclusion}
This study has explored the secrecy performance of FBL transmissions in terms of AIL for a multi-antenna wiretap communication system in the absence of Eve's instantaneous CSI. In particular, we have unveiled the relationship between the AIL in the FBL regime and the SOP for the IBL case and found that the latter is the saddle-point approximation of the former for a specific redundancy rate. We have further evaluated the AIL of the system with arbitrary precoding and fading channels and performed case studies to holistically investigate the impact of key system parameters on the AIL performance. Our findings suggest that increasing blocklength reduces the system AIL and sacrificing a small level of secrecy can lead to large gains in reliability. In this sense, AIL can be considered to be a measure of quality of security (QoSec) to be traded off with other measures of transmission quality \cite{chorti2022context}. We have also formulated and solved an AST optimization problem in terms of both blocklength and power allocation, obtaining important insights into the practical design of PLS with FBL transmissions. 
Future research will involve an exploration of the implications of FBL communication in meeting AIL requirements under imperfect CSI, while leveraging advanced optimization techniques such as deep reinforcement learning to enhance the efficiency of complex system designs.


\appendices
\numberwithin{equation}{section}
\makeatletter 
\newcommand{\section@cntformat}{Appendix \thesection:\ }
\makeatother

\section{\textcolor{black}{Proof of Proposition \ref{prop1}}}\label{AppendixA}
\textcolor{black}{To commence the proof, we rewrite \eqref{delta_bar} equivalently as
\begin{align}\label{delta_bar_CDF}
   \begin{split}
    \bar{\delta} =\int_{\R^+} &\left[1-F_{\gamma_e}(x)\right]\\
       &\times\frac{\partial }{\partial x}\Q\left(
       \stackrel{}{\underset{\text{g(x)}}{\underbrace{\sqrt{\frac{N}{V_e(x)}}\left[\log_2\left(\frac{1+\gamma_b}{1+x}\right)-{R_0}\right]}}}
          \right) dx,
   \end{split} 
\end{align}
where \eqref{delta_bar_CDF} is obtained using integration by parts. Now, keeping in mind that the first-order derivative of the Q-function is    $\frac{\partial \Q(x)}{\partial x} = -\frac{1}{\sqrt{2\pi}}\exp\left(-\frac{x^2}{2}\right)$, we apply the chain rule to obtain the derivative of the composite Q-function inside the integration as $\frac{\partial \Q}{\partial x} = f_1(x) f_2(x)$ where the functions $f_1$ and $f_2$ are defined over $0 <x \leq x_0$ as
\begin{align}
     f_1(x)&\treq{\frac{\partial \Q}{\partial g}}\Big|_{g(x)}\nonumber\\
     &=-\frac{1}{\sqrt{2\pi}}\exp\left(-\frac{N}{2V_e(x)}\left[\log_2\left(\frac{1-\gamma_b}{1+x}\right)-R_0\right]^2\right)
\end{align}
and
\begin{align}
    f_2(x)&\treq {\frac{\partial g}{\partial x}}\Big|_{x}\nonumber\\
    &=-\sqrt{\frac{N}{x(x+2)}}\left(1+\frac{\log_2\left(\frac{1+\gamma_b}{1+x}\right)-R_0}{x(x+2)\log_2\e}\right),
\end{align}
respectively. Accordingly, we can rewrite \eqref{delta_bar_CDF} as
\begin{align}\label{laplace}
      \bar{\delta} =  \int_{\R^+} \Psi(x) \e^{-N\Xi(x)} dx,
\end{align}
where the functions $\Psi(x)$ and $\Xi(x)$ are defined respectively as
\begin{align}\label{Psi_func}
    \Psi(x) = -\frac{1}{\sqrt{2\pi}}f_2(x)[1-F_{\gamma_e}(x)]
\end{align}
 and
\begin{align}\label{Xi_func}
    \Xi(x) = \frac{\left(\log_2\left(\frac{1+\gamma_b}{1+x}\right)-{R_0}\right)^2}{2V_e(x)}.
\end{align}
We note that $\Psi$ is a continuous function and always positive in its domain. Also, $\Xi$ is a twice-differentiable function. It is not difficult to show that the global minimum of $\Xi$ over $\R^+$ occurs at $x_0$ as $\Xi(x)$ is non-increasing with $x$. Since all the required conditions for the Laplace/saddle-point approximation method are satisfied, we can approximate \eqref{laplace} as
\begin{align} \label{delt1}
    \bar{\delta} &\approx \Psi(x_0) \e^{-N\Xi(x_0)} \sqrt{\frac{2\pi}{N\Xi''(x_0)}}.
\end{align}
We now determine the values of the functions in \eqref{delt1} to further simplify $\bar{\delta}$. By substituting $x_0$ into \eqref{Psi_func} and \eqref{Xi_func}, one can respectively obtain   
\begin{align}\label{psi_xi}
 \hspace{-3mm}   \Psi(x_0)=\sqrt{\frac{N}{2\pi x_0(x_0+2)}}[1-F_{\gamma_e}(x_0)]~\text{and}~\Xi(x_0)=0.
\end{align}
To calculate the second derivative of $\Xi$ at $x_0$, i.e., $\Xi^{''}(x_0)$ , we rewrite 
$\Xi(x) = \frac{(C_s(x)-R_0)^2}{2V_e(x)}$ where $C_s(x)=\log_2\left(\frac{1+\gamma_b}{1+x}\right)$ whose first derivative w.r.t. $x$ is $C'_s(x)=-\frac{\log_2 \e}{x+1}$. Then, we obtain $\Xi^{''}(x)$, after some tedious calculus, as
\begin{align}\label{Xi2d_func}
    \Xi^{''}(x)&=\frac{\left({V_e}(x) {C'_s}(x)-({C_s}(x)-{R_0}) {V'_e}(x)\right)^2}{{V_e}^3(x)}\nonumber\\
    &+\frac{{V_e}^2(x) ({C_s}(x)-{R_0}) C''_s(x)}{{V_e}^3(x)}\nonumber\\
    & - \frac{{V_e}(x) ({C_s}(x)-{R_0})^2 V''_e(x)}{2 {V_e}^3(x)}.
\end{align}
Substituting $x_0$ into \eqref{Xi2d_func} and using $C_s(x_0) = R_0$ leads to
\begin{align}\label{Xi2d_last}
    \Xi''(x_0) = \frac{C'^2_s(x_0)}{V_e(x_0)} = \frac{1}{x_0(x_0+2)}.
\end{align}
Therefore, using \eqref{psi_xi} and \eqref{Xi2d_last}, we can simplify \eqref{delt1} and obtain the approximate AIL expression given by \eqref{deltaApprox}, which completes the proof.}

\bibliographystyle{IEEEtran}
\bibliography{RefList}

\end{document}